\documentclass[a4paper,UKenglish, cleveref, autoref,thm-restate]{lipics-v2021}
\usepackage[utf8]{inputenc}
\usepackage{soul}

\nolinenumbers
\usepackage{amssymb,amsthm,amsmath}
\usepackage{xcolor,xspace}
\usepackage[linesnumbered, lined, boxed]{algorithm2e}
\usepackage{graphicx}
\usepackage{todonotes}
\usepackage{comment}
\usepackage{caption, subcaption}
\usepackage{complexity}
\usepackage{boxedminipage}

\newtheorem{rr}{Reduction Rule}

\newcommand{\fcn}{\operatorname{fcn}}

\newcommand{\SPQR}{SPR\xspace}

\newcommand{\yes}{\textsf{YES}}
\newcommand{\no}{\textsf{NO}}
\newcommand{\fpt}{\textsf{FPT}\xspace}

\newcommand{\problem}[3]{
        \begin{center}
                \begin{boxedminipage}{0.99\textwidth}
                        \textsc{#1}
                        
                        \vspace{2pt}
                        
                        \begin{tabular}{l p{0.8\textwidth}}
                                \textit{Instance:} & {#2}\\
                                
                                \textit{Question:} & {#3}
                        \end{tabular}
                \end{boxedminipage}
        \end{center}
}

\title{A Polynomial Kernel for Face Cover on Non-Embedded Planar Graphs}

\author{Thekla Hamm}{TU Eindhoven}{t.l.s.hamm@tue.nl}{0000-0002-4595-9982}{Supported by the Austrian Science Fund (J4651-N) during part of this work.}
\author{Sukanya Pandey}{RWTH Aachen}{sukanya.pandey@durham.ac.uk}{0000-0001-5728-1120}{}
\author{Krisztina Szil\'agyi}{Czech Technical University in Prague}{szilakri@fit.cvut.cz}{0000-0003-3570-0528}{Supported under the project Robotics and advanced industrial production (reg. no. CZ.02.01.01/00/22\_008/0004590).}
\authorrunning{T.\ Hamm, S.\ Pandey, K.\ Szil\'agyi}
\Copyright{Thekla Hamm, Sukanya Pandey, Krisztina Szil\'agyi}

\ccsdesc[500]{Theory of computation~Parameterized complexity and exact algorithms}
\ccsdesc[500]{Mathematics of computing~Graph algorithms}
\keywords{Kernelization, Planar Graphs, SPQR-tree}
\begin{document}
\maketitle
\begin{abstract}
    Given a planar graph, a subset of its vertices called terminals, and \(k \in \mathbb{N}\), the \textsc{Face Cover Number} problem asks whether the terminals lie on the boundaries of at most $k$ faces of some embedding of the input graph. When a plane graph is given in the input, the problem is known to have a polynomial kernel~\cite{GarneroST17}. In this paper, we present the first polynomial kernel for \textsc{Face Cover Number} when the input is a planar graph (without a fixed embedding). Our approach overcomes the challenge of not having a predefined set of face boundaries by building a kernel bottom-up on an SPR-tree while preserving the essential properties of the face cover along the way.
\end{abstract}

\clearpage

\section{Introduction}

Problems in network design have been extensively studied on the class of planar graphs. Many of these problems take as input a graph and a subset of its vertices distinguished as terminals, and the goal is to find an optimal subgraph connecting or separating the terminals that satisfies certain constraints. Even when the input is restricted to the class of planar graphs, a large number of these problems remain intractable if the number of terminals is superconstant~\cite{GareyJ77, VYGEN1995, Schwarzler09, Lynch75, DahlhausJPSY94, garg1997, kramer1984}. This motivates the quest for further restrictions on the input that may make the problems tractable.
Besides bounding the actual number of terminals, one can target weaker restrictions, which impose conditions on the placement of the terminals in the input planar graph. One such restriction could be to require that the terminals lie on the boundary of a single face of the planar input graph. Some well-known NP-hard problems like \textsc{Steiner Tree}, \textsc{Multiway Cut}, and \textsc{Shortest Disjoint Paths} can be solved in polynomial time~\cite{Erickson, Bern, Chen-Wu, BorradaileNZ15} in this case. It is therefore natural to wonder: \emph{What if the terminals lie on the boundaries of $k>1$ faces?}.

The \emph{terminal face cover number} of a planar graph is the minimum number of faces, over all embeddings, that jointly cover all the terminals. Parameterization by the terminal face cover number was first considered by Erickson et al.\ in their seminal paper~\cite{Erickson} to design an \textsf{XP} algorithm for \textsc{Steiner Tree}. Recently, it has received a lot of attention, be it for flow and cut problems~\cite{krauthgamer2020refined,Krauthgamer1,Chen-Wu, PandeyL22}, shortest path problems~\cite{ChenXu_shortestpaths, Frederickson91}, minimum non-crossing walks~\cite{EricksonN11}, and minimum Steiner tree~\cite{KNL20}. In this paper, we focus on the problem of computing the terminal face cover number. Formally, the problem is defined as follows:

\problem{Face Cover Number}{Planar graph $G$, $T\subseteq V(G)$, integer $k$}{Does there exist a set of faces of size at most $k$ in some embedding of $G$ that covers all the vertices of $T$?}

Bienstock and Monma~\cite{BM88} initiated the study of \textsc{Face Cover Number} and showed that it can be solved in \fpt time parameterized by $k$. Their algorithm runs in $c^{k}\cdot n$ time, for some constant $c>1$. They also studied a related problem where the input consists of a plane graph, that is, a planar graph with a fixed embedding. We call this problem \textsc{Embedded Face Cover Number}.
For this problem, they considered parameterization by the number of terminals $t$, and gave the first subexponential \fpt algorithm with a running time of $2^{\mathcal{O}(\sqrt{t})}$. 

A special case of \textsc{Embedded Face Cover Number} with $T=V(G)$, parameterized by the face cover number, was studied by Kloks et al.\ \cite{KloksLL02}. They showed that the problem can be solved in time $\mathcal{O}^*(c^{\sqrt{k}})$. Subsequently, their result was improved by Fernau and Juedes~\cite{FernauJ04}, whose algorithm runs in time $\mathcal{O}^*(2^{24.551\sqrt{k}})$ and further by Koutsonas and Thilikos~\cite{KoutsonasT11} who gave an algorithm with a running time of $\mathcal{O}^*(2^{10.1\sqrt{k}})$. 
This problem can be easily reduced to \textsc{Red-Blue Dominating Set}, a variant of \textsc{Dominating Set} on 2-colored graphs. All the aforementioned algorithms \emph{indirectly} find the minimum face cover of the input plane graph by solving \textsc{Red-Blue Dominating Set}. A direct \fpt algorithm for the problem was given by Abu-Khzam et al.\ \cite{Abu-KhzamFL08}. However, with a running time of $\mathcal{O}^*(4.6056^k)$, their algorithm is not subexponential in the parameter.

Given the fixed-parameter tractability of both the above problems, a natural line of inquiry is whether there exists a polynomial kernel for the problem. For \textsc{Embedded Face Cover Number}, Garnero et al.\ \cite{GarneroST17} showed that a linear kernel exists by reduction to \textsc{Red-Blue Dominating Set}. Before their work, a quadratic kernel was known to exist due to the result of Kloks et al.\ \cite{KloksLL02}. To date, it was not known whether \textsc{Face Cover Number} has a polynomial kernel. We answer this question in the affirmative by presenting a cubic kernel for the problem parameterized by the face cover number.
\begin{theorem}\label{thm:main}
 Given an instance $(G,T,k)$ of \textsc{Face Cover Number}, we can produce a kernel $(G',T',k')$ in polynomial time such that $|G'|+k' \in \mathcal{O}(k^3)$.
\end{theorem}

\vspace{1em}
\noindent\textbf{Overview of the algorithm.}\hspace{.5em}
Our approach is to perform dynamic programming over an \SPQR-tree of the given planar graph.
SPR-trees are classic data structures that have often been useful for algorithm design on planar graphs.
At a high level, an SPR-tree can be viewed as a tree-structured decomposition of a planar graph into its 3-connected subgraphs.
Every node of this tree-structure is associated with a graph called the \emph{skeleton} of that node. The graph induced at any node is obtained by replacing each of the so-called \emph{virtual edges} (edges of the skeleton that are not actual edges of the planar graph) by the graphs induced at the children of that node.
In this sense, endpoints of virtual edges can be viewed as an ``interface'' between the graph induced at a node and the graphs induced at its children. 
Each node of the SPR-tree has a type -- S, P or R -- depending on the structure of its skeleton.

With this in mind, we aim to construct a kernel for \textsc{Face Cover Number} in a leaves-to-root fashion along an SPR-tree of the input graph.
For this, a natural attempt at the dynamic programming step is to try to replace parts of the graph that correspond to an SPR-tree rooted at a child of that SPR-tree by an inductively assumed kernel.
However, this na\"ive approach comes with some technical obstacles: how do we replace a subgraph with a kernel in which we are not ensured that the interface to the subgraph remains?
On a more fundamental level, this simple idea is doomed to fail because while the face cover number of the kernel may be the same, it might interact differently with the rest of the graph than the subgraph we replace it with (see \Cref{fig:bad-replacement}).
Consequently, we need to strengthen what we inductively assume of kernels constructed so far and then also demonstrate in each dynamic programming step that these assumptions are satisfied after carrying out the reduction steps.
Crucially, our stronger assumptions allow us to preserve the way in which the face cover of the graph induced at any child node interacts with the graph induced at the parent node during kernelization.
We call a kernel that satisfies our stronger assumptions a \emph{nice kernel}, and our target is to dynamically compute such a kernel.
\begin{figure}[t]
    \centering
    \includegraphics[width=0.75\textwidth]{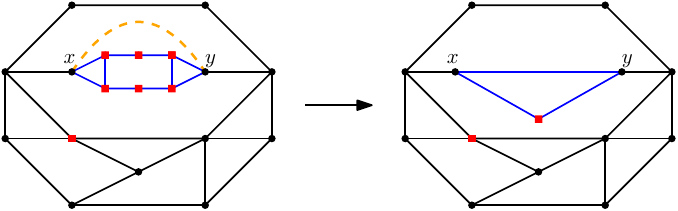}
    \caption{(left) A skeleton in black and orange with the orange dashed edge \(xy\) being a virtual edge to a child node. The graph that replaces this virtual edge is blue. Terminals are red squares. The blue graph has face cover number one (all its terminals lie on its rectangularly drawn face).
    Hence, a kernel for the blue graph could in principle be a triangle on \(x\), \(y\) and a terminal (this also has face cover number one).
    This would even still contain the interface (\(x,y\)) between the blue and black graph.
    (right) However, if we replace the blue graph with its possible triangular kernel, the face cover of the entire graph drops from two to one.}
    \label{fig:bad-replacement}
\end{figure}

While many parts of our dynamic programming procedure require us to devise reduction rules that are specific to the skeleton and hence the type of node we are considering (S, P or R), some subgraphs that are induced at the children of a node have sufficiently \emph{simple} face covers. Hence, we can replace them by constant-sized subgraphs and obtain a nice kernel in which we only need to handle subgraphs induced at children that have significantly more complex interactions with the skeleton of the parent. 

We refer to the simple subgraphs as \emph{terminal-free}, \emph{unproblematic} or \emph{semi-problematic} components depending on whether they contain terminals, and if so, whether they can be covered optimally using only faces that interact with the rest of the graph. We describe polynomial-time algorithms to recognize each of them.
We call the remaining complex subgraphs induced by children \emph{problematic}.
At each node type, the overall strategy is to bound the number of virtual edges that we need to consider by a polynomial in the prospective face cover number \(k\). Once we achieve that, we can then reduce the skeleton in a way that we can bound the number of its non-virtual edges in \(k\). Both of these are the steps that are the most specific to the node types.
Finally, we inductively replace the problematic virtual components by their kernels and any possible remaining non-problematic ones by their constant-size replacements -- this last step does not depend on the particular node type.

Concerning the node-type specific arguments, simple reduction rules like edge deletions and contractions can be used at S-nodes and P-nodes.
However, such simple reduction rules are difficult to control at R-nodes, which, of all the node types admit the structurally most complex skeletons, namely 3-connected planar graphs.
Deleting or contracting edges from a 3-connected planar graph does not, in general, preserve 3-connectivity.
This is of great relevance for face cover numbers because, before any reduction, the skeleton has a unique embedding, a property that we cannot easily preserve during reduction. 
For this reason, during the kernelization for R-nodes, we switch to the embedded setting, and formulate reduction rules that take the initial fixed embedding into account.
To transition back to the setting in which we consider a non-embedded graph, we re-establish 3-connectivity by what we call ``\emph{rigidizing}'' the embedded nice kernel.
While doing so, we take care not to change the sets of terminals of the nice kernel that occur together on a face, and hence obtain a non-embedded nice kernel at R-nodes.

\vspace{1em}
\noindent\textbf{Organization.} \hspace{.5em}
In \Cref{sec:prelims}, we give necessary definitions and notation. In \Cref{sec:virtcomptype}, we give formal definitions of \emph{terminal-free}, \emph{unproblematic} or \emph{semi-problematic} and \emph{problematic} components, and algorithms to distinguish them.
Then, in \Cref{sec:basickernel} we define the notion of a \emph{nice kernel} and describe the replacements of the components from \Cref{sec:virtcomptype}, which are not node-type specific.
In \Cref{sec:DP}, we turn to the actual formulation of the dynamic program.
We conclude in \Cref{sec:conclusion}.

\section{Preliminaries}
\label{sec:prelims}

\noindent\textbf{Graph terminology.}\hspace{.5em}
We follow standard graph terminology~\cite{Diestel}.
A \emph{graph} is a pair $G=(V, E)$, such that $E \subseteq {V\choose 2}$. The elements of $V$ are called the vertices of the graph and the elements of $E$ are called edges. The number of vertices of a graph is its order denoted by $n$. A \emph{subgraph} $G'= (V', E')$ of $G$, written as $G'\subseteq G$, is a graph such that $V'\subseteq V$ and $E'\subseteq E$. $G'$ is an \emph{induced subgraph} of $G$, if for all $x,y \in V'$, if $xy \in E$ then $xy \in E'$.

A \emph{path} is a non-empty graph $P=(V',E')$ of the form $V'=\{x_0, x_1, \ldots, x_k\}$ and $E'=\{x_0x_1, \ldots,$ $x_{k-1}x_k\}$, where all $x_i$ are distinct. The vertices $x_0$ and $x_k$ are called the \emph{endpoints} of $P$. An \emph{induced path} is a path in which all the vertices except the endpoints have degree two. The number of edges in a path is its \emph{length}.

A graph is \emph{connected} if there is a path between any two vertices in a graph, and \emph{disconnected} otherwise. A graph is \emph{biconnected} if there does not exist any vertex whose removal from the graph results in a disconnected graph. Similarly, a graph is called \emph{3-connected} if there does not exist a subset of vertices of size two whose removal disconnects the graph.


\vspace{1em}
\noindent\textbf{SPR-trees.}\hspace{.5em}
 An \emph{SPR-tree} is a rooted tree in which each node $t$ is associated with a multigraph called the \emph{skeleton} of the node. The nodes, and their skeletons, have one of three types:
 \begin{enumerate}
     \item {\bf S node:} These are nodes whose skeleton forms a cycle of three or more vertices. 
     \item {\bf P node:} These are nodes whose skeleton is a graph with exactly two vertices and at least three edges between them. Such a graph is called a dipole.
     \item {\bf R node:} These are nodes whose skeleton forms a 3-connected graph that is not a cycle or a dipole. 
 \end{enumerate}
 Each edge between two nodes \(t\) and \(t'\) of the SPR-tree is associated to one edge with an ordering of its endpoints in the skeleton of \(t\) and one edge with an ordering of its endpoints in the skeleton of \(t'\).
 Overall, the edge of the skeleton of each node is allowed to be associated to at most one edge of the SPR-tree in this way.
 These skeleton edges that are associated to SPR-tree edges are called \emph{virtual} edges.
 The edges of the skeleton of any node that are not virtual are called \emph{real}.
 For the subtree rooted at a node \(t\) of the SPR-tree the graph \emph{induced} by it is inductively defined as follows.
 \begin{itemize}
     \item if \(t\) is a leaf, the graph induced by the subtree rooted at \(t\) is the skeleton of \(t\).
     \item if \(t\) has children \(t_1, \dotsc, t_c\), the graph induced by the subtree rooted at \(t\) is the graph arising from the skeleton of \(t\) by doing the following for each \(i \in [c]\):
     Insert the graph induced by the subtree rooted at \(t_i\); identify the first and second endpoint of the virtual edge in the skeleton of \(t_i\) that is associated to \(tt_i\) with the first and second endpoint respectively of the virtual edge in the skeleton of \(t\) that is associated to \(tt_i\); and lastly, delete the virtual edges corresponding to \(tt_i\).
     Intuitively speaking, the graph induced by the subtrees rooted at \(t_i\) are `glued-in' in place of the respective virtual edges of the skeleton of \(t\).
 \end{itemize}
 We say that an SPR-tree is an SPR-tree of the graph that it induces.
 \begin{lemma}[Hopcroft and Tarjan~\cite{HopcroftTarjan73}]
     An \SPQR-tree of a biconnected planar graph can be constructed in linear time.
 \end{lemma}
For a node $t$ of an arbitrary fixed \SPQR-tree of \(G\), we denote the subgraph induced by the subtree rooted at $t$ by $G_t$. We call the subgraph induced by a subtree rooted at a child of \(t\) a \emph{virtual component} of \(G_t\).
Let \(c_1,c_2\) be the endpoints of the virtual edges associated to \(tt'\) where \(t'\) is a child of \(t\).
We call \(c_1\) and \(c_2\) the \emph{corners} of the virtual component of \(G_t\) associated with \(t'\). We call the edge $c_1c_2$ the \emph{corner edge} of \(G_{t'}\).
We refer to the virtual component of \(G_t\) associated with \(t'\) plus its corner edge as the \emph{enhancement} of that virtual component. In an embedding of an enhancement, we call the two faces containing the corner edge \emph{external faces} and the remaining faces \emph{internal faces}.
If \(G_t\) itself has no corners, i.e.\ \(t\) is the root of the SPR-tree, then \(G_t\) is equal to its enhancement (we pick an arbitrary edge of $G_t$ and regard it as the corner edge and its endpoints as corners).

\vspace{1em}
\noindent\textbf{Kernelization.}\hspace{.5em}
A kernelization for a parameterized problem $\Pi \in \Sigma^* \times \mathbb{N}$ is an algorithm that takes as input an instance $(I,k)$ of $\Pi$, runs in time $poly(|I|+k)$ and outputs an instance $(I',k')$ of $\Pi$ such that: (1) $(I,k)$ is a \yes-instance if and only if $(I',k')$ is a \yes-instance (2) $|I'| \leq f(k)$, for some computable function $f$, and (3) $k' \leq g(k)$, for some computable function $g$.
The output $(I',k')$ of a kernelization is referred to as a \emph{kernel}. We say that $(I', k')$ is a \emph{polynomial kernel} if $|I'|+k' \leq poly(k)$.

\vspace{1em}
\noindent\textbf{Embedded graphs.}\hspace{.5em}
A \emph{drawing} of a connected planar graph $G=(V, E)$ is defined as a mapping of the vertices in $V$ to points in $\mathbb{R}^2$ and edges to simple curves between the points corresponding to its endpoints, such that: (1) distinct edges have distinct pair of endpoints, and (2) the interior of an edge does not contain any vertex or any point of another edge.
Given a planar drawing, the (clockwise) cyclic order of edges incident on every vertex is \emph{fixed}. The set of circular orderings of edges around each vertex is called a \emph{rotation system}. A \emph{(combinatorial) embedding} is an equivalence class of planar drawings defined by their rotation systems.
We frequently identify vertices and edges with their drawings and embeddings with an arbitrary drawing from its represented equivalence class.

The restrictions of embeddings and drawings to a subgraph of a graph are called \emph{subembeddings} and \emph{subdrawings} respectively. Let $G$ be a plane graph and $t$ a node in its SPR-decomposition. \emph{Flipping} $G_t$ is defined by reversing the the rotation system restricted to \(G_t\) (i.e.\ for the corner vertices, we only reverse edges in $G_t$).

For every drawing $G$ of a planar graph, the inclusion-maximal regions of $\mathbb{R}^2 \setminus G$ are called \emph{faces}.
The unique unbounded face is the \emph{outer face}, all other faces are \emph{inner faces}.
Let $f$ be a face. 
We say that a face $f$ \emph{covers} a vertex $v$ if the drawing of $v$ lies on the boundary of $f$.
For the two external faces of the enhancement $G_t$, we define their corresponding faces in $G$ as the faces whose boundaries restricted to $G_t$ are equal. Given a face cover $F$ of $G$, we define the induced face cover $F_t$ of the enhancement of $G_t$ as the set of internal faces of $G_t$ that are in $F$ together with the external faces corresponding to a face in $F$.
The \emph{embedded face cover number} of an embedded graph \(G\) with terminals \(T \subseteq V(G)\) is the minimum number of faces needed  to cover \(T\).
The \emph{face cover number} of a graph $G$ with terminals \(T \subseteq V(G)\) ($\fcn(G)=\fcn(G,T)$) is the minimum number of faces over all embeddings of $G$ that cover the terminals of $G$.

\subsection*{One-connected case}
By standard arguments in face cover literature which we repeat below for self-containedness, we can reduce to considering a biconnected input graph. Also, note that we can reduce to connected graphs as follows. If $G$ is disconnected, we kernelize every connected component separately, and then we check if there are more than $k$ components that need at least 2 faces for the face cover (using the existing FPT algorithm from~\cite{BM88} for $k=1$). If there are, we have a \no-instance.

The following definition will be useful to show that we can reduce to the case of a biconnected input graph.
A \emph{cut vertex} in $G$ is a vertex $v$ such that $G-v$ has two or more connected components. A \emph{block} of graph $G$ is a maximal biconnected induced subgraph. A block is \emph{trivial} if it contains exactly one edge and \emph{non-trivial} otherwise. A \emph{block-cut tree} of a graph $G$ has a vertex for each block and cut vertex in $G$. We denote the block-cut tree by $\mathcal{B}$. There exists an edge in $\mathcal{B}$ for a pair of nodes corresponding to a block and cut vertex in $G$ if and only if the cut vertex belongs to the block.
\begin{lemma}[Hopcroft and Tarjan~\cite{HopcroftT73}]\label{lem:bctree}
A block-cut tree of a graph can be computed in linear time.   
\end{lemma}

We now give the reduction to the biconnected case.
\begin{restatable}{lemma}{oneconn}
   If we can compute $\fcn(G, T)$ for every 2-connected planar graph $G$ and every $T\subseteq V(G)$, then we can compute $\fcn(G, T)$ for every planar graph $G$ and \(T \subseteq V(G)\).
\end{restatable}
\begin{proof}
If $G$ is not 2-connected, consider the block-cut tree $\mathcal{B}$ of $G$. We will compute the \textsc{Face Cover Number} inductively, starting from the leaves. Let $v$ be a cutvertex with children $B_1,\dots, B_t$ and let $G^v$ be the subgraph of $G$ that corresponds to the subtree of $\mathcal{B}$ rooted at $v$. We distinguish two cases:
\begin{itemize}
    \item If $v$ is a terminal: we set $\fcn(G^v, T\cap G^v)=\sum \fcn(B_i, T\cap B_i)-(t-1)$, since we know that in each $B_i$, the vertex $v$ will be covered, and we can ensure that the face used to cover it is the outer face of $B_i$.
    \item If $v$ is not a terminal: for each $i\in [t]$, we compute $k_i=\fcn(B_i, T\cap B_i)$ and $k'_i=\fcn(B_i, (T\cap B_i)\cup \{v\})$. If there is exists $i>1$ such that $k_i=k'_i$, we can ensure that the outer face of $B_i$ covers $v$. Thus adding $v$ to $T$ does not increase the face cover number. We add $v$ to $T$ and use the above argument to compute $\fcn(G^v, T\cap G^v)$. Otherwise, we set $\fcn(G^v, T\cap G^v)=\sum \fcn(B_i, T\cap B_i)$.    \qedhere
\end{itemize}
\end{proof}

\section{Types of Virtual Components}\label{sec:virtcomptype}
To facilitate constructing the kernel, we will divide virtual components into 4 groups. The terminal-free, unproblematic and semi-problematic will be replaced by constant-sized gadgets, while the problematic ones will eventually be replaced by inductively constructed kernels.

\begin{definition}[Types of virtual components]
We distinguish different types of virtual components.
\begin{enumerate}
    \item Virtual components that have no terminals except possibly the corners -- we call such virtual components \emph{terminal-free}. 
\item Virtual components that have a face cover consisting of exactly one external face -- we call such virtual components \emph{unproblematic}.
\item Virtual components that are not unproblematic, have a face cover consisting of both external faces but also have face cover number one -- we call such virtual components \emph{semi-problematic}.
\item  All other virtual components which we call \emph{problematic}.
\end{enumerate}
\end{definition}

For unproblematic virtual components, we say that an embedding witnesses their unproblematicness if in that embedding one of the external faces covers all terminals. For semi-problematic components, we say that an embedding witnesses their semi-problematicness if the face cover number of that embedding is one, but both the external faces can jointly cover all the terminals.

We first argue that the types of virtual components can be recognized in polynomial time and then present kernelization steps for such virtual components in \Cref{sec:basickernel}.
Obviously, terminal-free virtual components can be recognized in linear time and to recognize problematic components, one merely needs to recognize unproblematic and semi-problematic components.
For this, we can give dynamic programs along the SPR-subtree.

\begin{restatable}{lemma}{unprobrec}
    Deciding whether a virtual component is unproblematic is in \P.
\end{restatable}
\begin{proof}
    This can be done by dynamic programming along the sub-SPR-tree of that virtual component.

    In the base case, this sub-tree consists of only one node (which is not a P-node).
    If this node is an S-node, the virtual instance is unproblematic if and only if all the terminals lie on a single path between its corners.
    If this node is an R-node, then the virtual component's embedding and face cycles are uniquely determined. One can check in polynomial time whether all the terminals lie on a single face cycle that contains the edge between the corners, which is exactly when such a virtual instance is unproblematic.

    We now turn to the dynamic programming step and distinguish cases based on the type of the root \(r\) of the sub-SPR-tree of the virtual component.
    
    \noindent\textbf{S-node.}\hspace{.5em}
    The virtual component is unproblematic if and only if the child virtual components associated with edges on one of the two paths in the skeleton of \(r\) between the corners of the virtual instance contain no terminals other than the corners and all child virtual components are unproblematic.
    The forward implication of this statement can be seen by embedding each child virtual component in a way that witnesses their unproblematicness and flipping these embeddings in a way that all terminals lie on the same face.
    For the backward direction, notice that if both paths in the skeleton of \(r\) between the corners of the virtual instance contain terminals other than the corners, then no face has all these terminals and the edge between the corners on its boundary, no matter what embedding choices we make.
    Moreover, if a child virtual component is not unproblematic, then it cannot be a subgraph of an unproblematic virtual component by the definition.

    \vspace{1em}
    \noindent\textbf{P-node.}\hspace{.5em}
    The virtual component is unproblematic if and only if at most one child virtual component contains terminals other than the corners of the virtual component and all child virtual components are unproblematic.
    The forward implication of this statement can be seen by embedding each child virtual component in a way that witnesses their unproblematicness and ordering and flipping the embedding of the (possibly) one child virtual component containing terminals in a way that all the terminals lie on the same face as the edge between the corners of the virtual component.
    For the backward direction, notice that if two child virtual components contain terminals other than the corners, then no face has all these terminals and the edge between the corners on its boundary, no matter what embedding choices we make.
    Moreover, if a child virtual component is not unproblematic, then it cannot be a subgraph of an unproblematic virtual component by the definition of being unproblematic.

    \vspace{1em}
    \noindent\textbf{R-node.}\hspace{.5em}
    The virtual component is unproblematic if and only if all child virtual components containing terminals are linked via virtual edges on a single face cycle which also contains the edge between the corners and all child virtual components are unproblematic.
    The forward implication of this statement can be seen by embedding each child virtual component in a way that witnesses their unproblematicness and flipping these embeddings in a way that all terminals lie on the same face as the edge between the corners.
    For the backward direction, notice if not all virtual edges linking child virtual components that contain terminals lie on a face cycle with the edge between the corners, then no face has all these terminals and the edge between the corners on its boundary, no matter what embedding choices we make.
    Moreover, if a child virtual component is not unproblematic, then it cannot be a subgraph of an unproblematic virtual component by the definition of being unproblematic.
    
    Checking the condition for each node type can be done in polynomial time.
\end{proof}

\begin{restatable}{lemma}{semiprobrec}
    Deciding whether a virtual component is semi-problematic is in \P.
\end{restatable}
\begin{proof}
    We can check in polynomial time, whether a graph has face cover number one.
    It remains to check whether all terminals can be covered in the enhancement of the virtual instance using both faces incident on the edge between the corners of the virtual instance.
    We call this property being \emph{external-face coverable} (\emph{EFC}).
    This can be done in a similar way as checking unproblematicness by a dynamic program along the sub-SPR-tree of the virtual component.

    In the base case, this sub-tree consists of only one node (which is not a P-node).
    If this node is an S-node, the virtual instance is semi-problematic.
    If this node is an R-node, then the virtual component's embedding and face cycles are uniquely determined and one can check in polynomial time whether all the terminals lie on the union of both face cycles that contain the edge between the corners.

    We now turn to the dynamic programming step and distinguish cases based on the type of the root \(r\) of the sub-SPR-tree of the virtual component.

    \noindent\textbf{S-node.} \hspace{.5em}
    The virtual component is EFC if and only if all child virtual components are EFC.
    The forward implication of this statement can be seen by embedding each child virtual component in a way that witnesses that they are EFC and flipping these embeddings in a way that all terminals lie on a face with the edge between the corners on its boundary.
    For the backward direction, notice that if a child virtual component is not EFC, then some of its terminals cannot lie on a face containing the edge between the corners on its boundary, no matter what embedding choices we make. 

    \vspace{1em}
    \noindent\textbf{P-node.} \hspace{.5em}
    The virtual component is EFC if and only if at most two child virtual component contain terminals other than the corners of the virtual component and all child virtual components are unproblematic.
    The forward implication of this statement can be seen by embedding each child virtual component in a way that witnesses their unproblematicness and ordering and flipping the embedding of the at most two child virtual components containing terminals in a way that all terminals lie on a face that has the edge between the corners of the virtual component on its boundary.
    For the backward direction, notice that if three child virtual components contain terminals other than the corners, then some terminal is not on the same face as the edge between the corners, no matter what embedding choices we make.
    Moreover, if a child virtual component is not unproblematic, then some of its terminals cannot lie on a face containing the edge between the corners on its boundary, no matter what embedding choices we make.

    \vspace{1em}
    \noindent\textbf{R-node.} \hspace{.5em}
    The virtual component is EFC if and only if all child virtual components containing terminals are linked via virtual edges in the union of the face cycles which also contain the edge between the corners and all child virtual components are unproblematic.
    The forward implication of this statement can be seen by embedding each child virtual component in a way that witnesses their unproblematicness and flipping these embeddings in a way that all the terminals lie on a face with the edge between the corners.
    For the backward direction, notice if not all virtual edges linking child virtual components that contain terminals lie on a face cycle with the edge between the corners, no face has all these terminals and the edge between the corners on its boundary, no matter what embedding choices we make.
    Moreover, if a child virtual component is not unproblematic, then at least one of its terminals cannot lie on the same face as the edge between the corners of the virtual component because two edges share at most one face in a non-trivial 3-connected graph by \Cref{prop:3-conn:three-shared-vertex-bound}.

    Checking the condition for each node type can be done in polynomial time.
\end{proof}

\section{Nice Kernels and Basic Virtual Component Replacement}
\label{sec:basickernel}
 In this section, we describe a kernelization procedure that inductively assumes that we have special kernels, which we refer to as \emph{small nice kernels} of the enhancement of each virtual component of the graph induced at the node of the SPR-tree under consideration. 

\begin{definition}
Let $(G,T)$ be an instance of \textsc{Face Cover Number}. Given a node \(t\) of the SPR-tree of \(G\) consider its enhancement \(G_t\) with corner vertices $c_1,c_2$.
We define $\fcn_0(G_t, T),\, \fcn_1(G_t,T),\, \fcn_2(G_t,T)$ as the sizes of the smallest face cover of the enhancement of $G_t$ such that they contain exactly 0, 1, 2 external faces respectively. Let $K$ be a graph with terminal set $T' \subseteq V(K)$. We say that $K$ is a \emph{nice kernel} of $G_t$ if the following conditions are fulfilled:

\begin{itemize}
    \item[\textbf{K1}] $K$ contains vertices $c_1,c_2$ and the edge $c_1c_2$.
    \item[\textbf{K2}] For $i\in\{0,1,2\}$, we have that $\fcn_i(G_t, T\cap V(G_t))=\fcn_i(K, T')\leq k$ or that $\fcn_i(G_t, T\cap V(G_t))>k$ and $\fcn_i(K, T')>k$.
    \item[\textbf{K3}] For every $C\subseteq \{c_1,c_2\}$, we have that
    $\fcn_0(G_t, (T\cap V(G_t))\setminus C)=\fcn_0(K, T'\setminus C)\leq k$ or that $\fcn_0(G_t, (T \cap V(G_t))\setminus C)>k$ and $\fcn_0(K, T'\setminus C)>k$,
\end{itemize}
A nice kernel is called \emph{small} if
\begin{itemize}
    \item[\textbf{K4}] There is some constant \(c\) (that does not depend on \(K\)) such that \(c\cdot |V(K)|^{\frac{1}{3}}\) is the minimum number of internal faces used in any minimum face cover of \(G_t\).
\end{itemize}
\end{definition}
The condition (K1) ensures that we can construct the kernel inductively, i.e. that in $G_t$, we can replace a virtual component corresponding to its child by a nice kernel for it. Intuitively, we want $G_t$ and $K$ to behave the same, and in particular interact with the rest of the graph in the same way.  
The conditions (K2) and (K3) say that if the face cover sizes of $G_t$ and $K$ differ, both of the face covers are too big to lead to a global solution. 
One important thing to note is that the corner vertices $c_1,c_2$ might be covered by from the ``outside'' of $G_t$. Thus we need to compare face covers of $G_t$ and $K$ covering all terminals except for the corners. 
In (K3), we only compare the values of $\fcn_0$ as the values of $\fcn_1$ and $\fcn_2$ are equal by (K2) (because the corners are covered by the corresponding external faces). 

Ultimately, we will design a leaves-to-root dynamic program along the SPR-tree to obtain small nice kernels for each \(G_t\) with \(t\) being a node of an SPR-tree of \(G\).
We can verify that a small nice kernel at the root of the SPR-tree is in fact a polynomial kernel.
\begin{lemma}
\label{lem:kernel-root}
    Let \(r\) be the root of the considered SPR-tree for \(G\).
    A small nice kernel \(K\) of \(G_r\) is a polynomial-sized kernel for \(G\) as an instance of \textsc{Face Cover Number}.
\end{lemma}
\begin{proof}
    By the fact that for \(r\), \(G_r\) has no corners and \(G=G_r\) coincides with its enhancement, (K1)--(K3) are obsolete or equivalent to expressing instance equivalence of \(G\) and \(K\) as instances of \textsc{Face Cover Number}.

    Moreover, by (K4) there is a constant \(c\) such that \(|V(K)|\leq \frac{1}{c}(k - 2)^3\) where \(k\) is the face cover number of \(G\) (and hence also of \(K\)).
    This shows that \(K\) is a cubic kernel for \(G\) as an instance of \textsc{Face Cover Number}.
\end{proof}

 The next reduction rules replace ``simple'' virtual components by constant-sized graphs.
 \begin{rr}\label{rr:terminal-free}
     Replace every virtual component without terminals by an edge between its corners.
 \end{rr}
 \begin{rr}\label{rr:unproblematic}
     Replace every unproblematic component by a $P_3$ whose endpoints are identified with the corners and the degree-two vertex is a terminal. 
 \end{rr}
 
 We call the above graph a \emph{terminal-subdivided edge}.
 \begin{rr}\label{rr:semi-problematic}
     Replace every semi-problematic component by the gadget depicted in \Cref{fig:half-problematic} by connecting one corner to the vertex $a$ of the gadget by an edge and the other corner to the vertex $b$.
 \end{rr}
 
 \begin{figure}[t]
    \centering
    \includegraphics{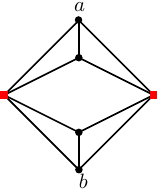}
    \caption{The gadget for semi-problematic virtual components. Red squares are terminals. Note that if \(a\) and \(b\) are fixed to be on the outer face, then the face boundaries are unique.}
    \label{fig:half-problematic}
\end{figure}

\begin{figure}[t]
    \centering
    \includegraphics[width=0.75\linewidth]{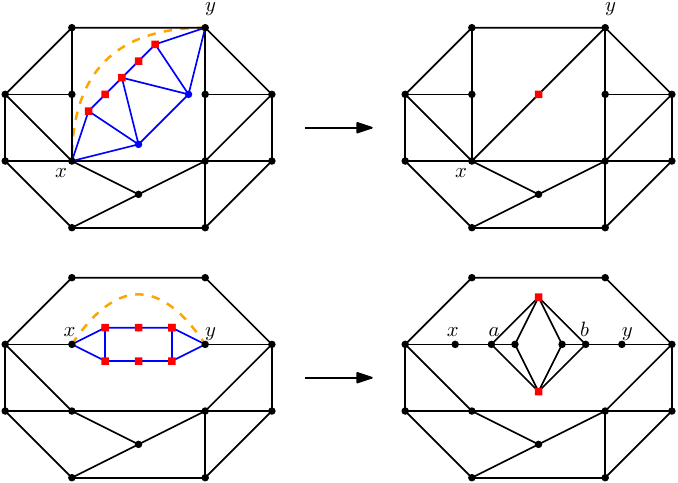}
    \caption{Exemplary replacements made in \Cref{rr:unproblematic} (top) and \Cref{rr:semi-problematic} (bottom). The virtual edges are shown as orange dashed curves, their corresponding virtual components are shown in blue in the left figures and their corresponding replacements are shown on the right.}
    \label{fig:replacements}
\end{figure}

As defined above, let $t$ be a node of the \SPQR-tree of $G$ and let $G_t$ be its enhancement.
\begin{restatable}{lemma}{basic}
    \label{lem:basic-steps}
    Let \(K\) and \(T'\) be the graph that arises from \(G_t\) by applying the described replacements (\Cref{rr:terminal-free}, \Cref{rr:unproblematic} and \Cref{rr:semi-problematic}) and the set of terminals in it respectively. Then \(K\) is a nice kernel of \(G_t\).
\end{restatable}
\begin{proof}
Obviously (K1) is satisfied, as the corners of \(G_t\) are not among the vertices being replaced.

For (K2) consider any \(i \in [3]\), and fix an embedding of \(G_t\) in which \(\fcn_i(G_t)\) can be achieved and let \(F\) be a face cover doing so.
    Without loss of generality, \(F\) contains no internal faces of the subembeddings of terminal-free virtual components (such faces can simply be omitted from \(F\)).
    Similarly, without loss of generality, \(F\) contains no internal faces of the subembeddings of unproblematic virtual components; otherwise, all unproblematic virtual component can be redrawn such that only one of its external faces
    is needed to cover all its terminals.
    If one of the external faces is also the external face of \(G_t\), the other one is not and we are free to include either into \(F\) depending on \(i\).
    In fact, we can assume the subembeddings of unproblematic virtual components to witness their unproblematicness by the same argument.
    Lastly, along a similar line of arguments, without loss of generality, \(F\) either contains one internal face of any semi-problematic component or both its external faces, and we can assume the subembeddings of semi-problematic virtual components to witness their semi-problematicness.

    Now let \(F'\) arise from \(F\) by replacing any part of a face boundary in it that contained an edge from (i) a terminal-free virtual component by the edge between its corners, (ii) an unproblematic virtual component by the path replacing that unproblematic component, and (iii) each outer face of the subembedding of a semi-problematic virtual component by each of the terminal-subdivided edges in the graph drawn in Figure~\ref{fig:half-problematic} replacing that virtual component.
    As a final step, for faces in \(F\) that are internal faces of semi-problematic virtual components, we instead include a face into \(F'\) consisting of the unique internal face in an embedding according to \ref{fig:half-problematic} that contains both new terminals of the gadget.
    We embed \(G'\) based on the embedding of \(G\) as follows.
    Embed each edge replacing a terminal-free virtual component in \(G'\) along the embedding of one of the paths between the corners in that component, each terminal-subdivided edge replacing an unproblematic virtual component in \(G'\) along the embedding of one of the paths containing a terminal between the corners in that component, and each small graph replacing a semi-problematic virtual component along the embedding of that component in \(G\) where we embed that small graph itself as in \ref{fig:half-problematic}.

    In this embedding of \(G'\) it is straightforward to verify that \(F'\) is a face cover, its size is equal to that of \(F\) and it contains the same number of external faces of \(G_t\).

    Conversely, fix an embedding of \(K\) in which \(\fcn_i(K)\) can be achieved and let \(F'\) be a face cover doing so.
    Without loss of generality \(F'\) only contains inner faces of a gadget replacing a semi-problematic virtual component that contain both new terminals of that gadget; otherwise we can simply replace replace an inner face of a gadget in \(F'\) that does not contain both terminals of that gadget by the inner face of the gadget that does.
    We can replace the subembeddings of replacements of terminal-free, unproblematic and semi-problematic virtual components by \(\epsilon\)-thin embeddings of their originals.
    In the case of unproblematic virtual components we choose embeddings that witness their unproblematicness and are flipped with the single face covering all terminals in the direction of a face in \(F'\) (such a face must exist because unproblematic components are replaced by terminal-subdivided edges and the corresponding terminal must be covered by \(F'\)).
    In the case of semi-problematic virtual components we choose embeddings that witness their semi-problematicness.
    We obtain \(F\) from \(F'\) by first replacing any face in \(F'\) that is an inner face of a semi-problematic component gadget by the internal face covering all terminal of that semi-problematic component.
    For each semi-problematic component for which no inner face of its gadget is contained in \(F'\) without loss of generality both faces in the fixed drawing of \(K\) that contain the corners of that virtual component have to be in \(F'\).
    To obtain \(F\) we replace them by the faces in the fixed drawing of \(G_t\) that satisfy the same, i.e.\ the two faces that contain both corners of the considered semi-problematic component.
    
    It is straightforward to verify that \(F\) is a face cover in the described embedding of \(G\), its size is equal to that of \(F'\) and it contains the same number of external faces of \(K\).

    For (K3) the same arguments as for (K2) work irrespective of which corners of \(G_t\) are considered terminals.
\end{proof}

Next, we show that inductively replacing nice kernels for problematic virtual components yields a nice kernel. 
\begin{lemma}\label{lem:basic-problematic}
   Let $C$ be a problematic component of \(G_t\) and $\tilde{C}$ be a nice kernel for $C$. Let $\tilde{G_t}$ be the graph obtained from $G_t$ by replacing $C$ by $\tilde{C}$. Then $\tilde{G_t}$ is a nice kernel for $G_t$.
\end{lemma}
\begin{proof}
    Firstly, note that the replacement does not affect the skeleton of $t$, so the corner vertices and corner edge of $G_t$ are preserved, i.e. (K1) is satisfied. Let $c_1,c_2$ be the corner vertices of $C$ and let $C'$ be the enhancement of $C$.

    Let us now show (K2). Let $i\in \{0,1,2\}$. 
    Fix an embedding $\mathcal{F}$ of $G_t$ witnessing $\fcn_i(G_t, T)$ and let $F$ be the corresponding face cover. We aim to construct a face cover $\tilde{F}$ of $\tilde{G_t}$ from $F$. 
    Informally, we partition $F$ into two sets, $F_C$ and $F_U$, where $F_U$ is the set of faces unaffected by the replacement of $C$ by $\tilde{C}$, and $F_C$ contains the remaining faces of $F$.   
    Using the fact that $\tilde{C}$ is a nice kernel for $C$, we construct a face cover $\tilde{F}_{C'}$ of $\tilde{C}$. The face cover $\tilde{F}$ will contain the faces in $F_U$ and the faces corresponding to $\tilde{F}_{C'}$.
    
    Formally, we define $F_U=\{f\in F:\: f\cap V(C)\subseteq \{c_1,c_2\} \}$. In other words, $F_U$ contains faces in $F$ that are disjoint from $C$ except possibly for its corner vertices. Let $F_C=F\setminus F_U$. The faces in $F_C$ correspond to faces of $C'$: those that are completely contained in $C$ correspond to internal faces of $C$ (with the same set of vertices) and those that contain vertices outside of $C$ correspond to the external faces of $C'$. Let $F_{C'}$ be the set of faces of $C'$ corresponding to $F_C$.
    Note that $F_{C'}$ is a face cover of $C'$ covering all terminals in $T\cap V(C)$ except possibly for $c_1,c_2$ (as these vertices may be covered by faces in $F_U$).
    Let $j\in \{0,1,2\}$ be the number of faces in $F_{C'}$ corresponding to external faces of $C'$. We distinguish two cases depending on the value of $j$:
    
    If $j\in\{1,2\}$: in this case, all terminals of $C'$ are covered by a face in $F_{C'}$. By an exchange argument, it is easy to see that $F_{C'}$ is optimal, i.e. it witnesses $\fcn_j(C, T\cap V(C))$, so we have $\fcn_j(C, T\cap V(C))=|F_{C'}|=|F_C|$. Fix an embedding of $\tilde{C}$ witnessing $\fcn_j(\tilde{C}, T'\cap V(\tilde{C}))$ and let $\tilde{F}_{C'}$ be the corresponding face cover. We set $T_1=T$, $T_1'=T'$.
    
    If $j=0$: in this case, $c_1$ and $c_2$ might not be covered by a face in $F_{C'}$. Let $T_1, T_1'$ be sets of terminals obtained by removing the corner vertices not covered by a face in $F_C$ from $T$ and $T'$ respectively. Similarly, we have that $\fcn_0(C, T_1'\cap V(C))=|F_C|$. Fix an embedding of $\tilde{C}$ witnessing $\fcn_0(\tilde{C}, T_1'\cap V(\tilde{C}))$ and let $\tilde{F}_{C'}$ be the corresponding face cover.
    
    We embed $\tilde{G_t}$ by using the above embedding of $\tilde{C}$ and leaving the remaining part of $\tilde{G_t}$ the same as in $\mathcal{F}$. If $j=1$, we embed $\tilde{C}$ inside $\tilde{G_t}$ such that the external face of $\tilde{C}$ that belongs to $F_{C'}$ corresponds to the face in $F_C$ containing $c_1,c_2$. If $j\in\{0,2\}$ we can embed $\tilde{C}$ with an arbitrary choice of external faces (i.e. the way it is ``flipped'' inside $\tilde{G_t}$ is irrelevant). Now we have obtained an embedding of $\tilde{G_t}$.
    
    We define $\tilde{F}_C$ as the set of faces of $\tilde{G_t}$ corresponding to $\tilde{F}_{C'}$. Note that the faces in $F_U$ are faces of $\tilde{G_t}$ (since they do not contain vertices from $V(C)\setminus\{c_1,c_2\}$) and that $F_U\cap \tilde{F}_{C}=\emptyset$.
    
    We claim that $\tilde{F}=F_U\cup \tilde{F}_C$ is a face cover of $\tilde{G_t}$. Clearly, all terminals except $c_1,c_2$ are covered: those in $V(\tilde{C})$ are covered by a face in $\tilde{F}_C$, and the others are covered by a face in $F_U$. The vertices $c_1$ and $c_2$ (if they are terminals) are covered by a face in $\tilde{F}_C$ if they were covered by a face in $F_C$, and otherwise they are covered by a face in $F_U$. Thus $\tilde{F}$ is a face cover of $\tilde{G_t}$. It remains to show the size bounds.

    If $|F|\leq k$, then we have $\fcn_j(C, T_1\cap V(C))=|F_C|\leq k$, so $|\tilde{F}_{C}|=\fcn_j(\tilde{C}, T_1'\cap V(\tilde{C}))=\fcn_j(C, T_1\cap V(C))$ by (K2) and (K3). Thus $|\tilde{F}_{C}|=|F_C|$, so $|\tilde{F}|=|F|$. If $|F|>k$, we distinguish two cases. If $|F_C|>k$, then $|\tilde{F}_{C}|>k$, so $|\tilde{F}|>k$. If $|F_C|\leq k$, then $|\tilde{F}_{C}|=|F_C|$, so $|\tilde{F}|=|F|>k$.  

    To show (K3), we remove one or both corner vertices from $T$ and $T'$ and proceed the same as the case $j=0$.
\end{proof}

\section{Node-Type Specific Kernelization Steps}
\label{sec:DP}
In the following subsections, we describe the node-specific reduction rules that will allow us to obtain small nice kernels at each node type.
By \Cref{lem:kernel-root}, this implies a polynomial kernel for the entire graph.
For each of these sections, consider a node \(t\) of the SPR-tree of \(G\).
\subsection{R-node}
\label{sec:rnode}
It is well-known that any 3-connected planar graph, in particular the skeleton of an R-node, has a unique embedding up to homeomorphism~\cite{Brinkmann21}.
Our strategy is to fix this embedding of the skeleton of an R-node, kernelize it based on this fixed embedding, possibly losing 3-connectivity, and then rigidize the resulting kernel to regain control over the faces that contain terminals.

We first make the following general observation about 3-connected embedded graphs.
\begin{restatable}{proposition}{threesharedvertex}
\label{prop:3-conn:three-shared-vertex-bound}
    In any 3-connected planar graph with at least four vertices, no three vertices share more than one face.
\end{restatable}
\begin{proof}
    Let \(H\) be a planar 3-connected graph with at least four vertices.
    We identify \(H\) with its unique embedding.
    Assume for contradiction that there are three distinct vertices \(u,v,w\) that lie on the shared boundary of two faces \(f\) and \(f'\).
    We fix such \(u,v,w\) with the property that along a (counterclockwise) traversal of \(f\), they appear in the order \(u,v,w\).
    By planarity, this implies that along \(f'\) the (counterclockwise) order of appearance is \(u,w,v\).

    Because \(f\) and \(f'\) are faces and because of the assumed order of appearance along \(f\), we can draw a closed curve that only intersects \(H\) in \(u\) and \(v\) and otherwise is entirely contained in the interiors of \(f\) and \(f'\) such that this curve separates \(w\) from the traversal \(t\) of \(f\) between \(u\) and \(v\).
    If there was any vertex \(x\) other than \(u\) or \(v\) on \(t\), then this would mean that there can be no \(x\)-\(w\)-path in \(H - u - v\) contradicting 3-connectivity.
    Hence, \(uv\) is an edge on \(f\).
    By symmetric arguments we get that \(vw\) and \(wu\) are edges on the boundary of \(f\), and \(uv\), \(vw\) and \(wu\) are edges on the boundary of \(f'\).
    Overall this implies that \(u,v,w\) form a triangle that is disconnected from an existing fourth vertex, a contradiction to 3-connectivity.
\end{proof}


Proposition~\ref{prop:3-conn:three-shared-vertex-bound} can easily be adapted to when some edges are subdivided by terminals, in particular as is the case in the replacement for unproblematic virtual components in \Cref{sec:basickernel}.
For ease of notation, we refer to graphs that may arise from 3-connected planar graphs by replacing some edges by terminal-subdivided edges as \emph{almost 3-connected planar} and to the graph that arises from the skeleton of \(G\) in this way as \emph{extended skeleton}.

\begin{restatable}{proposition}{threesharedvertexp}
\label{prop:3-conn+:three-shared-vertex-bound}
    In any embedding of any almost 3-connected planar graph with at least seven vertices, no four vertices share more than one face.
\end{restatable}

\begin{proof}
    Let \(H\) be a almost 3-connected planar graph with at least seven vertices.
    We identify \(H\) with its unique embedding.
    Assume for contradiction that there are four distinct vertices that lie on the shared boundary of two faces.
    
    By construction, whenever two faces share a degree-2 vertex, they also have to share both its neighbors.
    Hence, if two faces share more than three vertices, among these there are three vertices that are not degree-2, or two degree-2 vertices implying the existence of three shared vertices that are not degree-2.

    This means that we can undo the subdivided-edge replacements by which \(H\) arose from a 3-connected planar graph \(H'\) and arrive at a 3-connected planar graph in which two faces share at least three vertices.
    Assume for contradiction that \(H'\) has less than four vertices.
    Then \(H'\) can have at most three edges and hence \(H\) can have at most six vertices, a contradiction.
    Thus we can apply Proposition~\ref{prop:3-conn:three-shared-vertex-bound} to \(H'\) to obtain a contradiction.
\end{proof}

Let us now actually focus on \(G_t\) where \(t\) is an R-node of the SPR-tree of \(G\).
For problematic virtual components, we make the following simple observation.
\begin{restatable}{observation}{fewprobcomps}
    If \(G_t\) has more than \(k\) problematic virtual components, then it has no face cover of size at most \(k\).
\end{restatable}
\begin{proof}
    To cover the terminals in a problematic virtual component, at least one internal face of a subembedding of that virtual component is necessary.
    Such faces cannot be shared among distinct virtual components.
%
\end{proof}

We can also bound the number of semi-problematic virtual edges that we need to consider.
\begin{restatable}{proposition}{fewsemiprobcomps}
\label{prop:3-conn:semiprob-face-bound}
    If there is a face of the skeleton of \(G_t\) with more than \(k\) semi-problematic virtual edges on its boundary, then there is no face cover of size at most \(k\) for this embedding.
\end{restatable}
\begin{proof}
    To cover the terminals in a semi-problematic virtual component, at least one face private to a subembedding of that virtual component or both its external faces (faces incident to both corners of the virtual component) are necessary.
    By Proposition~\ref{prop:3-conn:three-shared-vertex-bound}, no pair of virtual components corresponding to virtual edges that lie on a single face of the skeleton can share another face in the skeleton.
    This means to cover the terminals of more than \(k\) semi-problematic virtual components whose corresponding virtual edges lie on a single face of the skeleton, at least \(k+1\) pairwise distinct faces are necessary.
\end{proof}

\begin{restatable}{corollary}{actuallyfewsemiprobcomps}
    If there are more than \(k^2\) semi-problematic virtual edges in the skeleton of \(G_t\), then there is no face cover of size at most \(k\).
\end{restatable}

\begin{proof}
    Assume for contradiction that there is a face cover of \(G_t\) of size at most \(k\).
    By pigeon hole principle, at least one of these faces has to cover terminals from more than \(k\) semi-problematic virtual components.
    This face has to correspond to a face of the skeleton with more than \(k\) semi-problematic virtual edges on its boundary contradicting \Cref{prop:3-conn:semiprob-face-bound}.
\end{proof}

Finally, we reduce the number of unproblematic virtual edges and in fact, the total number of terminals in the skeleton.
For this, we first make a general statement about almost 3-connected planar graphs.
\begin{restatable}{proposition}{heavyface}
    \label{prop:3-conn+:heavy-face}
    In the embedding of an almost 3-connected planar graph, a face with more than \(3k\) terminals on its boundary has to be in any face cover of size at most \(k\).
\end{restatable}
\begin{proof}
    Assume for contradiction that there is a face cover that does not contain some such face.
    Then, at most three terminals on the boundary of such a face can be covered at a time by any other face due to Proposition~\ref{prop:3-conn+:three-shared-vertex-bound}.
    This means that at least \(k + 1\) faces are necessary to cover the terminals on the considered face alone, a contradiction.
\end{proof}

Together with the safeness of the basic replacement step for unproblematic virtual components (cf.\ the proof of \Cref{lem:basic-steps}), this immediately can be used to justify the following reduction rule which we refer to as \emph{terminal-heavy face reduction}.
\begin{restatable}[Safeness of terminal-heavy face reduction]{lemma}{safethfr}
    Replacing all unproblematic virtual components by terminal-subdivided edges and then for each face of the extended skeleton with more than \(3k + 1\) terminals, turning all but an arbitrary set of \(3k + 1\) of these terminals into non-terminals yields a nice kernel.
\end{restatable}
\begin{proof}
    Being a nice kernel after replacing the unproblematic virtual components is ensured by Lemma~\ref{lem:basic-steps}.

    Further, (K1) trivially holds.
    It remains to argue (K2) and (K3).
    
    Obviously, any face cover of the enhancement of \(G_t\) covers every subset of its terminals, irrespective of which external faces it must contain or which corners are allowed to be left uncovered.
    This shows that all numbers we consider in (K2) and (K3) do not increase when modifying \(G_t\) as described in the lemma.

    Now we show that any face cover covering the subset of terminals which we leave, also covers the terminals which we turned into non-terminals, thereby implying the remaining inequalities in (K2) and (K3).
    
    Consider a terminal \(v\) which is covered by a face cover in \(G_t\) before modification and which we turned into a non-terminal.
    By construction, it lies on a face of the extended skeleton which has \(3k+1\) terminals remaining and by Proposition~\ref{prop:3-conn+:heavy-face} such a face has to be included in the face cover of the instance we obtain after transforming some terminals into non-terminals.
    Hence \(v\) is also covered by such a face cover.
\end{proof}

\begin{restatable}{proposition}{terminalbound}
\label{prop:terminalbound}
    If after exhaustive application of terminal-heavy face reduction for \(G_t\), there are more than \(3k^2 + k\) terminals, then there is no face cover of size at most \(k\).
\end{restatable}
\begin{proof}
    Assume for contradiction that there is a face cover of \(G_t\) of size at most \(k\).
    By the pigeon hole principle, at least one of these faces has to cover more than \(3k + 1\) terminals.
    This contradicts the fact that terminal-heavy face reduction was exhaustively applied.
\end{proof}

Assume from now on that terminal-heavy face reduction was exhaustively applied to \(G_t\); in particular, we will no longer explicitly mention it as a precondition in the following theorem statements.
We have now bounded the number of terminals and edges in the extended skeleton, which we will want to replace by inductively assumed kernels.
It remains to reduce the parts of the skeleton that contain neither.
During this modification of the extended skeleton, for the first time, we might lose the uniqueness of its embedding.
We will fix its embedding and take it into account in the following modifications.
For this, we explicitly speak about nice kernels of \textsc{Embedded Face Cover Number}, rather than nice kernels for \textsc{Face Cover Number} i.e.\ \(\fcn\) gets replaced by its fixed-embedding analogue in (K1)--(K3).

We can bound the number of faces containing at least two of the following objects: terminals, corners or endpoints of virtual edges. We refer to any such vertex as \emph{interesting}.
\begin{restatable}{proposition}{fewtfaces}
\label{prop:Rnode-few2faces}
    The number of faces with at least two interesting vertices in the skeleton of \(G_t\) is bounded above by \(\mathcal{O}(k^2)\).
\end{restatable}
\begin{proof}
    To count the number of such faces, consider the auxiliary graph $F=(V_F, E_F)$. Let $V_F$ be the set of interesting vertices of $G_t$. For every face $f$ of the extended skeleton of $G_t$ containing two or more interesting vertices, we pick two arbitrary interesting vertices $u, v \in V_F$ that lie on the boundary of $f$, and add the edge $uv$ to $E_F$. Furthermore, $uv$ is embedded inside $f$. As there is exactly one edge of $E_F$ contained in any face of $G_t$, $F$ is a plane graph. Therefore, by Euler's formula $|E_F|=\mathcal{O}(|V_F|)=\mathcal{O}(k^2)$.
    Since the number of faces of $G_t$ containing at least two interesting vertices is $|E_F|$, we get that the number of such faces is $\mathcal{O}(k^2)$.
\end{proof}

\begin{figure}
    \centering
    \includegraphics[width=0.4\linewidth,page=1]{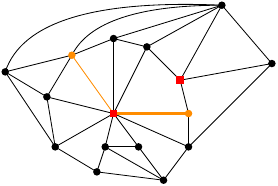}
    \hfil
    \includegraphics[width=0.4\linewidth,page=2]{rnodefigs.pdf}
    \caption{Example of the result of boring edge removal (right) applied to a graph (left) in which terminals are red squares, and virtual edges and their endpoints are orange (the thick one being the corner edge).} 
    \label{fig:Rnode-ber}
\end{figure}

\begin{figure}
    \centering
    \includegraphics[width=0.4\linewidth,page=3]{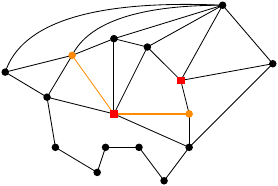}
    \hfil
    \includegraphics[width=0.4\linewidth,page=4]{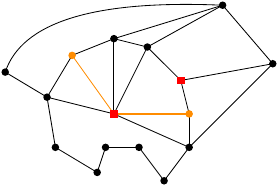}
    \caption{Example of the result of private face merging around the left terminal (left) and the result of exhaustive private face merging (right) applied to the graph from \Cref{fig:Rnode-ber}.}
    \label{fig:Rnode-pfm}
\end{figure}

\begin{rr}[Boring edge removal]\label{rr:Rnode-ber}
    Delete edges of the extended skeleton that are only on boundaries of faces that do not contain interesting vertices (see \Cref{fig:Rnode-ber}).
\end{rr}
Note that these edges in the reduction rule above are uniquely determined because the embedding of the extended skeleton is fixed.

Next, we consider each face \(f\) of the skeleton of \(G\) that has some arbitrary fixed terminal, corner or endpoint of a virtual edge \(x\) on its boundary.
\begin{rr}[Private face merging]\label{rr:Rnode-pfm}
    If there exists a face $f$ in the extended skeleton that has at most one fixed interesting vertex, say $x$, on its boundary, and is incident on another face $f'$ containing $x$ but no other interesting vertex, delete the shared boundary of $f$ and $f'$ connected to $x$ (see \Cref{fig:Rnode-pfm}).
\end{rr}

\begin{restatable}[Safeness of private face merging]{proposition}{safepfm}
    \label{prop:Rnode-facemergingsafe}
    Let \(\tilde{G}\) be the graph obtained from \(G_t\) by exhaustively applying \Cref{rr:Rnode-ber} and \Cref{rr:Rnode-pfm}.
    For an arbitrary embedding of \(G_t\), \(\tilde{G}\) embedded according to the restriction of the embedding of \(G_t\) to \(\tilde{G}\) is a nice kernel of \(G_t\) for \textsc{Embedded Face Cover Number}.
\end{restatable}
\begin{proof}
    Boring edge removal and private face merging does not change the sets of terminals that together occur on any face, hence ensuring (K2) and (K3).
    Also, no corner is removed ensuring (K1).
\end{proof}

Boring edge removal and private face merging allow us to reduce the problem to a setting in which the number of faces of the skeleton is bounded.
\begin{restatable}{lemma}{fewskelfaces}
    \label{lem:Rnode-fewskelfaces}
    Let \(\tilde{S}\) be the graph obtained from the skeleton of \(G_t\) by exhaustively applying \Cref{rr:Rnode-ber} and \Cref{rr:Rnode-pfm}, embedded according to the restriction of the unique embedding of the skeleton of \(G_t\) to \(\tilde{S}\).
    \(\tilde{S}\) has at most \(\mathcal{O}(k^3)\) faces.
\end{restatable}
\begin{proof}
    By \Cref{prop:Rnode-few2faces}, the number of faces of the skeleton of \(G_t\) with at least two interesting vertices \(\mathcal{O}(k^2)\).
    This number does not increase by boring edge removal or private face merging.

    Firstly, we argue that any interesting vertex \(x\) can be contained in at most \(\mathcal{O}(k^2)\) face boundaries which contain no other interesting vertex:
    Due to exhaustive application of \Cref{rr:Rnode-pfm}, no two faces containing $x$ but no other interesting vertex can share edges on their boundary connected to \(x\). 
    Hence in the rotation around \(x\), any two such faces must be separated by a face which also contains another interesting vertex, of which there are only \(\mathcal{O}(k^2)\).
    Hence we can find a mapping from all faces that contain precisely one interesting vertex to faces that contain at least two interesting vertices such that each of the latter is mapped to by only vertices on its boundary.
    Because each face has at most \(\mathcal{O}(k)\) interesting vertices on its boundary, each element of the \(\mathcal{O}(k^2)\)-sized co-domain of this mapping is mapped to at most \(\mathcal{O}(k)\) times.
    This means that there are at most \(\mathcal{O}(k^3)\) faces with exactly one interesting vertex on their boundary.

    Let us now bound the number of faces \(f\) of \(\tilde{S}\) that have no interesting vertex on their boundary.
    Each edge on the boundary \(f\) has to also be on the boundary of a face that contains an interesting vertex on its boundary, otherwise it would have been deleted during boring edge removal.
    As in the proof of \Cref{prop:Rnode-few2faces}, consider the auxiliary graph $F=(V_F, E_F)$. $V_F$ contains a vertex for each face of $\tilde{S}$ that contains at least one interesting vertex. For every face $f$ of $\tilde{S}$ that contains no interesting vertex, we pick two arbitrary faces containing interesting vertices adjacent to $f$ and add to $E_F$ an edge between their corresponding vertices. Note that for every such face $f$, $E_F$ contains exactly one edge and the faces without interesting vertices are disjoint from each other. Therefore, $F$ is a planar graph and by Euler's formula $|E_F|= \mathcal{O}(k^3)$. Hence, the number of faces of $\tilde{S}$ which are adjacent to at least two faces containing interesting vertices is $\mathcal{O}(k^3)$.

    It remains to bound the number of faces \(f_1, \ldots f_j\) of \(\tilde{S}\) without interesting vertices on their boundaries all of whose boundaries are also shared with a single face \(f'\), which contains an interesting vertex on its boundary.
    This can only happen if each \(f_i\) is nested inside \(f'\).
    Because the skeleton of \(G_t\) is almost 3-connected, there must have been at least three vertex-disjoint paths connecting the boundary of \(f_i\) to the part of the boundary of \(f'\) in which it nests.
    These paths partition \(f'\) into three subfaces.
    If any of these subfaces contains no interesting vertices, then part of the boundary of \(f_i\) would have been deleted during boring edge removal; a contradiction.
    Hence, all of these subfaces contain an interesting vertex and these must be pairwise distinct by the vertex-disjointness of the paths.
    Then the paths would still be contained in \(\tilde{S}\); a contradiction to \(f'\) being a face of \(\tilde{S}\).
\end{proof}

So far, we have shown an upper bound only on the number of faces. To obtain a small number of vertices rather than faces, we apply the following reduction rule.

\begin{rr}[Boring edge contraction] \label{rr:Rnode-bec}
  Replace each induced path of length at least two in the skeleton, whose interior is free from interesting vertices, by a single edge. Further, we delete vertices of degree one that are not interesting vertices.
\end{rr} 

Let $\tilde{G}$ be the graph obtained after applying \Cref{rr:Rnode-ber} and \Cref{rr:Rnode-pfm} to \(G_t\) embedded in an arbitrary way. Then let $\tilde{G}'$ be the embedded graph obtained after exhaustively applying \Cref{rr:Rnode-bec}. The sets of terminals occurring together on any face of $\tilde{G}$ do not change in $\tilde{G}'$. Hence, this is safe with respect to \textsc{Embedded Face Cover Number} by the same argument as \Cref{prop:Rnode-facemergingsafe}.
\begin{proposition}
\label{prop:Rnode-embeddedskeletonkernel}
    Let \(\tilde{G}\) be the graph obtained from \(G_t\) by exhaustively applying boring edge removal, private face merging and then boring edge contraction.
    For an arbitrary embedding of \(G_t\), let \(\tilde{G}\) be embedded according to the restriction of the embedding of \(G_t\) to \(\tilde{G}\) in which edges that replace paths trace these paths at an \(\varepsilon\)-distance. Then, $\tilde{G}$ is a nice kernel of \(G_t\) for \textsc{Embedded Face Cover Number}.
\end{proposition}

Finally, we are at a point where we can start bounding the number of vertices.
\begin{restatable}{lemma}{skelsize}
     Let \(\tilde{S}\) be the graph obtained from the skeleton of \(G_t\) by exhaustively applying boring edge removal, private face merging and then boring edge contraction.
     \(|V(\tilde{S})| \in \mathcal{O}(k^3)\).
\end{restatable}
\begin{proof}
    Notice that all vertices of degree at most two in \(\tilde{S}\) are interesting vertices as otherwise they would have been removed during boring edge contraction.
    We already know that we can bound the number of such vertices by \(\mathcal{O}(k^2)\).

    To bound the number of vertices with degree at least three, we consider the embedded graph \(\tilde{S}'\) arising from \(\tilde{S}\), in which we replace every induced path with more than two edges by a single edge.
    The vertices of this graph are precisely those of \(\tilde{S}\) that have degree at least three, and the minimum vertex degree in \(\tilde{S}'\) is three.

    This means that \(|E(\tilde{S})'| \geq \frac{3}{2}|V(\tilde{S}')|\).
    Plugging this into Euler's formula for \(\tilde{S}'\) we get that \(|V(\tilde{S}')| \leq 2f - 4\) where \(f\) is the number of faces of \(\tilde{S}'\).
    The number of faces of \(\tilde{S}'\) is the same as the number of faces of \(\tilde{S}\) which by \Cref{lem:Rnode-fewskelfaces} is in \(\mathcal{O}(k^3)\).

    Overall, this implies the lemma.
\end{proof}

We will now re-establish uniqueness of the embedding of our kernelized skeleton by \emph{rigidizing} (i.e.\ adding structures to the embedded graph that render it 3-connected).
\begin{definition}
    Let \(H\) be an embedded planar graph and \(U \subseteq V(H)\).
    The \emph{rigidization} of \(H\) around \(U\) is the embedded graph \(H^{\otimes U}\) arising from \(H\) in the following way.
    For each edge \(uv \in E(H)\) where \(\{u,v\} \cap U \neq \emptyset\) introduce two new copies of \(uv\), drawn at an \(\varepsilon\)-distance from \(uv\).
    Subdivide each of \(uv\) and its copies by \(|\{u,v\} \cap U|\) many vertices, which we refer to as \emph{subdivision vertices} at \(\varepsilon\)-distance from \(\{u,v\} \cap U\).
    Introduce the unique planar cycle on the subdivision vertices at \(\varepsilon\)-distance of each vertex in \(U\).
\end{definition}
See \Cref{fig:rigid} for an example of the above definition.
\begin{figure}
    \centering
    \includegraphics[width=0.4\linewidth,page=6]{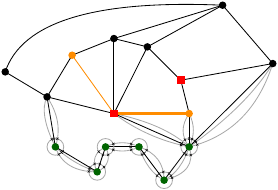}
    \caption{\(G^{\otimes U}\) with \(G\) from \Cref{fig:Rnode-pfm} (right) and \(U\) the set of green vertices. Edges that are not subdivisions of original edges are colored gray, subdivision vertices are indicated as crosses.
    Notice that only one vertex of \(U\) would remain after boring edge contraction.}
    \label{fig:rigid}
\end{figure}
Rigidization does not significantly increase the size of a graph:
\begin{restatable}{observation}{smallrigid}
    Let \(H\) be an embedded planar graph and \(U \subseteq V(H)\).
    \(|V(H^{\otimes U})| \in \mathcal{O}(|V(H)|)\).
\end{restatable}
\begin{proof}
    For each edge of \(H\), we introduce at most six new vertices in \(H^{\otimes U}\) (up to two for each of the three copies of an edge in \(H\)).
    Since \(|E(H)| \in \mathcal{O}(|V(H)|)\) as \(H\) is planar, this implies the statement.
\end{proof}

For us, the intended purpose of rigidization is to fix an embedding, which is formally supported by the following lemma.
\begin{restatable}{lemma}{tconnrigid}
\label{lem:Rnode-rigidization}
    Let \(H\) be an arbitrary embedded connected planar graph, and \(U\) be the set of vertices that are contained in a separator of size at most two of \(H\).
    \(H^{\otimes U}\) is 3-connected.
\end{restatable}
\begin{proof}
    Assume for contradiction that \(X \subseteq V(H^{\otimes U})\) is a separator of size at most two.
    \(X\) cannot consist of vertices that also exist in \(H\) but do not form a separator there.
    Hence, \(X\) consists of cutvertices or 2-separator vertices of \(H\) or subdivision vertices.

    Any path that contains in its interior a vertex \(u \in U\), i.e.\ a cutvertex or a 2-separator vertex of \(H\), can be rerouted\footnote{we always also allow short-cutting if the rerouting option intersects \(P\)} in \(H^{\otimes U}\) via each of the paths on the cycle at an \(\varepsilon\)-distance around \(u\) between the predecessor and the successor of \(u\).
    Obstructing both these rerouting options requires two vertices.
    Hence, if present in \(X\), \(u\) can be removed from \(X\) without changing the fact that \(X\) is a separator in \(H^{\otimes U}\).
    Hence, without loss of generality, \(X\) consists only of subdivision vertices.
    
    Any path \(P\) that contains in its interior a subdivision vertex \(w\) that is on the \(\varepsilon\)-distance cycle around \(u \in U\) can also be rerouted in \(H^{\otimes U}\) in two vertex-disjoint ways:
    if in \(P\), \(w\) is between \(u\) and another vertex \(w'\) on the \(\varepsilon\)-distance cycle around \(u\), then there is the direct edge \(w'u\) and the path consisting of the \(\varepsilon\)-distance cycle around \(u\) without \(ww'\).
    If in \(P\), \(w\) is between two other vertices \(w'\) and \(w''\)  on the \(\varepsilon\)-distance cycle around \(u\), then it can be rerouted via \(w'u\) and \(w''u\) or via the path consisting of the \(\varepsilon\)-distance cycle around \(u\) minus \(ww'\) and \(ww''\).
    Otherwise, \(w\) has a neighbor on \(P\) that is reached via an edge that arose by possibly copying and subdividing an edge \(uv \in E(H)\).
    If the other neighbor of \(w\) is \(u\), we can reroute via each of the two paths corresponding to the other copies of \(uv\).
    Otherwise, we can reroute the path corresponding to one of the other copies of \(uv\) and \(u\) and the path corresponding to the last remaining copy of \(uv\) and the \(\varepsilon\)-distance cycle around \(u\) without \(w\).
    
    In any case, we can avoid \(w\) by two vertex-disjoint rerouting options the obstruction of which requires two vertices.
    Hence, if present in \(X\),\(w\) can be removed from \(X\) without changing the fact that \(X\) is a separator in \(H^{\otimes U}\).

    This means that \(X\) is empty; a contradiction to \(H\) being connected.
\end{proof}

It is straightforward to verify that the graph \(\tilde{S}\) obtained from the skeleton of \(G_t\) by exhaustively applying boring edge removal, private face merging and then boring edge contraction is connected.
Hence \(\tilde{S}^{\otimes U}\) where \(U\) is the set of all cutvertices and vertices in 2-separators of \(\tilde{S}\) is planar and 3-connected by \Cref{lem:Rnode-rigidization} and hence has a unique embedding which allows us to return from considering \textsc{Embedded Face Cover Number} to \textsc{Face Cover Number} on abstract graphs.
That safeness is maintained in this step crucially relies on the fact that we never rigidize around interesting vertices.
\begin{restatable}{lemma}{nonembed}
\label{lem:Rnode-nospecialincut}
    Let \(\tilde{S}\) be obtained from the skeleton of \(G_t\) by exhaustively applying boring edge removal, private face merging and then boring edge contraction.
    No vertex separator of size at most two in \(\tilde{S}\) contains an interesting vertex.
\end{restatable}
\begin{proof}
    We first show that no cutvertex of \(\tilde{S}\) is an interesting vertex.
    Assume for contradiction, that \(x\) is such a cutvertex or in a 2-separator of \(\tilde{S}\) together with \(y \in V(\tilde{S}) \setminus \{x\}\).
    Let \(C_1\) and \(C_2\) be two arbitrary distinct connected components of \(\tilde{S} - x\) or \(\tilde{S} - \{x,y\}\) that are connected to \(x\).
    Because the skeleton of \(G_t\) is 3-connected, there is a path from a neighbor \(z_1\) of \(x\) in \(C_1\) to a neighbor \(z_2\) of \(x\) in \(C_2\) avoiding \(y\).
    We consider \(C_1\), \(C_2\) and such a path \(P\) whose embedding which together with \(x\) encloses an inclusion-minimal region of \(\mathbb{R}^2\) in the embedding of the skeleton of \(G_t\).
    All of the edges of \(P\) must have been deleted in the construction of \(\tilde{S}\) from the skeleton of \(G_t\).
    Edges were only deleted during boring edge removal, private face merging and boring edge contraction.

    All edges of \(P\) are together on a face with \(x\) by choosing \(P\) to be minimal with respect to the enclosed region with \(x\).
    Hence no edge of \(P\) is deleted during boring edge removal.
    During boring face merging, \(P\) can only be deleted together with \(xz_1\) or \(xz_2\), contradicting that these edges are still in \(\tilde{S}\).
    Finally, during boring edge contraction the only edges that are deleted are in 1-connected components of the graph arising from the skeleton of \(G_t\) after applying boring edge removal and private face merging.
    This is obviously not true of any edge in \(P\) as \(P + \{xz_1,xz_2\}\) forms a cycle.

    This means \(P\) is still in \(\tilde{S}\) contradicting the assumption that \(x\) is a cut vertex or \(\{x,y\}\) is a 2-separator.
\end{proof}

The following theorem essentially shows instance equivalence of what will be -- up to the replacement of virtual components -- the small nice kernel for \(G_t\).
\begin{restatable}{theorem}{nicekernel}
    Let \(\tilde{G}\) and \(\tilde{S}\) be the graphs obtained from \(G_t\) and its skeleton respectively by exhaustively applying boring edge removal, private face merging and then boring edge contraction.
    \(\tilde{G}^{\otimes U}\) where \(U\) is the set of all cut-vertices and vertices that are contained in a separator of size two in \(\tilde{S}\) is a nice kernel of \(G_t\) for \textsc{Face Cover Number}.
\end{restatable}
\begin{proof}
    (K1) is trivially satisfied as no interesting vertices are removed.
    We now argue (K2) and (K3).

    Fix an embedding of \(G_t\) in which its minimum face cover number using a specified number of external faces and covering a specified set of corners (we refer to this as \emph{generalized face cover number}) can be realized.
    By \Cref{prop:Rnode-embeddedskeletonkernel}, \(\tilde{G}\) is a nice kernel for \(G_t\) as an instance for \textsc{Embedded Face Cover Number}.
    By \Cref{lem:Rnode-rigidization}, \(\tilde{S}^{\otimes U}\) has a unique embedding which necessarily coincides with the unique one of the skeleton of \(G_t\) on its restriction to \(\tilde{S}\) where copies of edges are embedded at an \(\varepsilon\)-distance from the corresponding original edge and subdivision vertices are drawn in cycles at \(\varepsilon\)-distance around the corresponding vertices in \(U\).
    Further, by construction, replacing the virtual edges in this embedding by the embeddings of the corresponding virtual components of \(G_t\) results in an embedding of \(\tilde{G}^{\otimes U}\) in which the faces are the same as in the embedding of \(\tilde{G}\) apart from some vertices in \(U\) being replaced by subdivision vertices in some face cycles and there being new face cycles consisting only of vertices in \(U\) and subdivision vertices.
    By \Cref{lem:Rnode-nospecialincut}, this means that the sets of terminals, corners and virtual edge endpoints from the skeleton of \(G_t\) that occur together on face boundaries are the same in the fixed embedding of \(\tilde{G}\) and the described embedding of \(\tilde{G}^{\otimes U}\).
    This means that the abstract graph \(\tilde{G}^{\otimes U}\) has at most the same face cover number as the embedded graph \(\tilde{G}\), i.e.\ the generalized face cover number of \(G_t\) as an abstract graph.
    
    Conversely, we can fix an embedding of \(\tilde{G}^{\otimes U}\) in which its minimum generalized face cover number can be realized.
    This naturally induces an embedding of \(\tilde{G}\) by restriction and tracing the connections from cycles at \(\varepsilon\)-distances around \(U\) to the respective vertices in \(U\).
    By construction in this embedding, faces are the same as in the embedding of \(\tilde{G}^{\otimes U}\) apart from some subdivision vertices in some face cycles being replaced by vertices in \(U\) and faces consisting only of vertices in \(U\) and subdivision vertices being removed.
    By \Cref{lem:Rnode-nospecialincut}, this means that the sets of terminals, corners and virtual edge endpoints from the skeleton of \(G_t\) that occur together on face boundaries are the same in the fixed embedding of \(\tilde{G}^{\otimes U}\) and the described embedding of \(\tilde{G}\).
    We can appropriately extend the described embedding of \(\tilde{G}\) to \(G_t\) by extending the contained subembedding of \(\tilde{S}\) to the unique one of the skeleton of \(G_t\) and then replacing subgraphs attached at the endpoints of virtual edges by their corresponding virtual components drawn in a way that realizes the generalized face cover number variant relevant for that respective virtual component.
    Note that some generalized number face cover number variant has to be relevant as otherwise the virtual component or its replacement in the construction of \(\tilde{G}\) admit a smaller face cover with the same usage of external faces contradicting the minimality of the considered overall face cover.
    The subembedding of \(\tilde{S}\) is compatible with the unique embedding of the skeleton of \(G_t\) because the subembedding of \(\tilde{S}\) arises from the unique embedding of \(\tilde{S}^{\otimes U}\) which in turn arises from the restriction of the unique embedding of the skeleton of \(G_t\) by adding embeddings of edge copies at \(\varepsilon\)-distance from the original edges and cycles at \(\varepsilon\)-distance around vertices in \(U\).
    By the definitions of boring edge removal, private face merging and boring edge contraction, this extension does not subdivide faces of \(\tilde{G}\) in a way in which different sets of terminals, corners and virtual edge endpoints from the skeleton of \(G_t\) lie together in faces.
    This means that the abstract graph \(\tilde{G}^{\otimes U}\) has at least the same generalized face cover number as \(G_t\) as an abstract graph.
\end{proof}
As a final step, we replace all virtual components by their corresponding constant-size gadgets if they are semiproblematic and their inductively assumed kernels if they are problematic.
By \Cref{lem:basic-steps} this results in a nice kernel.
We can also show that this final nice kernel which we call \(K\) is also small.
\begin{restatable}{theorem}{smallnice}
    \(K\) as obtained by the process described in this subsection is small.
\end{restatable}
\begin{proof}
    We call internal faces that are internal faces that are also internal for some problematic virtual component \emph{deep} and other internal faces of \(G_t\) \emph{skeleton-internal}.
By construction, \(K\) arises from a 3-connected graph \(S\) with \(\mathcal{O}(k^3)\) vertices with at most \(\mathcal{O}(k^2)\) vertices per face, by replacing some of the edges by constant size gadgets and others by inductively assumed small nice kernels, where \(k\) is the maximum over the number of problematic and semi-problematic virtual components of \(G_t\) and the face cover number of the skeleton of \(G_t\).

We first consider the case that \(k\) is the number of problematic and semi-problematic virtual components of \(G_t\).
In \(G_t\), each semi-problematic virtual component requires the inclusion of at least one of its internal faces or one of its external faces in a way that these are pairwise distinct.
None of these are deep faces of problematic components.
Let \(d\) be such that \(x \leq d \cdot {k^\circ}^3\) where \(x\) is the number of vertices in \(S\) after the replacement of edges by constant-size components.
We can assume that \(c \geq d\).
The size of \(K\) is the sum of \(x\) and the sizes \(s_i\) of the small nice kernels of the problematic virtual components.
For sufficiently many (some constant depending on \(d\)) \(s_i\),
\begin{align}\left(\sum_i d + s_i\right)^{\frac{1}{3}} \leq \sum_i s_i^{\frac{1}{3}}.\label{eq:si}\end{align}
With this, we have
\[c \cdot \left(x + \sum_i s_i\right)^{\frac{1}{3}} \leq p_s + c \left(\sum_i d + s_i\right)^{\frac{1}{3}} \leq p_s + \sum_i c \cdot s_i^{\frac{1}{3}}\]
where \(p_s\) is the number of semi-problematic virtual components of \(G_t\).
This is a lower bound for the number of internal faces needed to cover \(G_t\).
If there are not sufficiently many \(s_i\), for (\ref{eq:si}) to hold, \(k \in \mathcal{O}(p_s)\) and we can immediately obtain a desired bound by setting \(x \leq c \cdot p_s^3\).

It remains to consider the case in which \(k\) is the face cover number of the skeleton of \(G_t\) which we divide into two subcases.

If the face cover number of the skeleton of \(G_t\) is at most four, \(S\) is of constant size this is also true after replacing semi-problematic virtual components.
That means the size of \(K\) is the sum of a constant and the sizes \(s_i\) of the small nice kernels of the problematic virtual components.
Deep faces do not interact with skeleton-internal faces or deep faces of other problematic virtual components.
Hence if they are necessary in any face cover of a virtual component, they are also needed in any face cover of \(G_t\) and hence \(\sum_i c \cdot s_i^{\frac{1}{3}} \geq c \cdot \left(\sum_i s_i\right)^{\frac{1}{3}}\) lower bounds the number of internal faces of \(G_t\) required in a face cover of \(G_t\).

Otherwise, \(k \leq 2k^\circ\) where \(k^\circ\) is the number of skeleton-internal faces required in a face cover of \(G_t\) and the number of problematic virtual components.
Hence \(|V(S)| \in \mathcal{O}({k^\circ}^3)\) and even when replacing edges by constant size gadgets this holds; we can assume \(c\) to have been chosen such that \(x \leq c \cdot {k^\circ}^3\) where \(x\) is the number of vertices in \(S\) after the replacement of edges by constant-size components.
The size of \(K\) is the sum of \(x\) and the sizes \(s_i\) of the small nice kernels of the problematic virtual components.
Again, deep faces do not interact with skeleton-internal faces or deep faces of other problematic virtual components.
Hence if they are necessary in any face cover of a virtual component, they are also needed in any face cover of \(G_t\).
Similarly, deep faces can also not replace skeleton-internal faces in any face cover of \(G_t\).
Hence \(c \cdot x^{\frac{1}{3}} + \sum_i c \cdot s_i^{\frac{1}{3}} \geq c \cdot \left(x + \sum_i s_i\right)^{\frac{1}{3}}\) lower bounds the number of internal faces of \(G_t\) required in a face cover of \(G_t\).
\end{proof}

\subsection{S-node}
\label{sec:snode}

In an S-node $t$, we first take care of the real edges in the skeleton and terminal-free virtual components as follows.

\begin{rr}\label{rr:S_real}
    If there are at least two terminals in the skeleton of the S-node that are not corner vertices, turn all but one into non-terminals. If there is a path of length at least two consisting of real edges, replace it by a single real edge.
\end{rr}
\begin{rr}\label{rr:S_terminal_free}
    Contract each terminal-free virtual component to a vertex.
\end{rr}

Our next task is to reduce the number of unproblematic components. If we consider only $\fcn_1$ and $\fcn_2$, it would suffice to keep only one unproblematic component (since all the other ones will be covered by the outer faces that are in the face cover). However, in case of $\fcn_0$, we do not use any of the outer faces in the face cover, so we have to use internal faces of the unproblematic components. In this case, having at least $k+1$ unproblematic components guarantees that $\fcn_0$ will be larger than $k$.

If there is an unproblematic component without internal faces (e.g.\ a terminal subdivided edge), we say that $\fcn_0=\infty$, and we an show that the same holds after applying the reduction rules (i.e.\ after applying reduction rules, there is no face cover without external faces). 

\begin{rr}\label{rr:S_unproblematic}
     If there are at least $k+1$ unproblematic virtual components, take all but $k+1$ of them and contract all their edges. For each remaining unproblematic component, if it has a face cover consisting of one internal face, replace it by a triangle two of whose vertices are corner vertices of the component, and the third vertex is a terminal. Otherwise, replace it by a terminal subdivided edge.
\end{rr}

For the same reason, we need to keep $k+1$ semi-problematic components:
\begin{rr}\label{rr:S_semi-problematic}
    If there are at least $k+1$ semi-problematic virtual components, take all but $k+1$ of them and contract all their edges. Replace the remaining semi-problematic components by using the gadget in \Cref{fig:half-problematic}.
\end{rr}

Finally, we replace each unreplaced problematic virtual component by its inductively assumed kernel for its enhancement restricted to the virtual component.
\begin{restatable}{lemma}{snicekernel}
 Consider an S-node $t$ and let $\tilde{G_t}$ be the graph obtained by applying the reduction rules exhaustively. Then $\tilde{G_t}$ is a nice kernel for $G_t$.
 \end{restatable}
\begin{proof}
    It is easy to see that after applying any reduction rule, the condition (K1) is satisfied, since we do not delete vertices and we do not remove the corner edge.
    
    Let us show that after applying \Cref{rr:S_real}, the resulting graph satisfies (K2) and (K3).
    Terminals in the skeleton of $t$ that are not corner vertices are covered if and only if we use at least one external face of $t$, i.e. covering one such terminal is equivalent to covering all such terminals. Replacing paths of real edges by single edges does not affect the size of the face cover or the number of external faces used, so the condition (K2) is satisfied. Since the corner vertices and the faces containing them are also unaffected by this reduction rule, (K3) also holds.
    
    Next, we show that after applying Reduction Rules~\ref{rr:S_terminal_free}, \ref{rr:S_unproblematic} and~\ref{rr:S_semi-problematic}, the resulting graph is a nice kernel. 
    Consider an embedding of $G_t$ witnessing $\fcn_0(G_t, T\cap V(G_t))$ let $F$ be the corresponding face cover. If there are terminals in the skeleton that are not corner vertices, we have no face cover without using external faces, and this remains true after applying the reduction rules. Thus we can assume that this is not the case. 
    
    Without loss of generality, we may assume that $F$ does not contain any internal faces of terminal-free virtual components and that it contains at most one internal face of each unproblematic and semi-problematic component. We construct a face cover $\tilde{F}$ of $\tilde{G_t}$ as follows. Note that $F$ contains no external faces, so we can partition it into $F_u, F_s$ and $F_p$, namely the sets of internal faces of unproblematic, semi-problematic and problematic components. 
     We define $\tilde{F_u}$ as the set of internal faces in the triangles constructed by applying \Cref{rr:S_unproblematic}, and $\tilde{F_s}$ the set of internal faces of the gadgets constructed by \Cref{rr:S_semi-problematic}. We define $\tilde{F_p}$ as the set of internal faces of kernels of problematic components covering all their terminals except possibly the corner vertices (i.e. covering precisely those corner vertices that are covered by $F_p$). Let $\tilde{F}=\tilde{F_u}\cup\tilde{F_s}\cup\tilde{F_p}$. It is easy to see that $\tilde{F}$ is a face cover of $\tilde{G_t}$ and that it does not contain either of the external faces.
     
     By construction, if $|F_u|\leq k$, then no unproblematic component was contracted, so $|F_u|=|\tilde{F_u}|$. On the other hand, if $|F_u|>k$, we have at least $k+1$ unproblematic components, and each of them requires one internal face, so $|\tilde{F_u}|>k$.
    The same holds for $F_s$ and $F_p$ (by the induction hypothesis). The proof of (K3) is analogous.

    Let $F$ now be the face cover of $G_t$ witnessing $\fcn_1(G_t, T\cap V(G_t))$. In this case, we can assume that $F$ does not contain any internal faces of unproblematic components (since it contains an external face that covers all their terminals). We define $F_s, F_p, \tilde{F_s}, \tilde{F_p}$ as above and obtain a face cover of $\tilde{G_s}$ consisting of $\tilde{F_s}, \tilde{F_p}$ and one external face.

    If $F$ is a face cover witnessing $\fcn_2(G_t, T\cap V(G_t))$, we can assume that it does not contain any internal faces of unproblematic and semi-problematic components and obtain a face cover of $\tilde{G_s}$ consisting of two external faces and $\tilde{F_p}$.  

    The other direction, i.e. constructing a face cover of $G_t$ from a face cover of $\tilde{G_t}$ is similar.
\end{proof} 

\begin{restatable}{lemma}{ssmall}
    After applying \Cref{rr:S_real}, \ref{rr:S_terminal_free}, \ref{rr:S_unproblematic}, \ref{rr:S_semi-problematic} exhaustively on $G_t$, we obtain a small nice kernel.
\end{restatable}
\begin{proof}
   Let $p_1,\dots,p_\ell$ be the sizes of kernels of problematic components in $\tilde{G_t}$ and let $u$ and $s$ be the number of unproblematic and semi-problematic components respectively.    
   Note that by construction, we have that $u\leq k+1$ and $s\leq k+1$. Each problematic and semi-problematic component is replaced by a gadget with at most 6 vertices, so $|V(\tilde{G_t})|\leq 6(u+s)+\sum p_i$. By the induction hypothesis and property (K4), we know that the $i$th problematic component requires at least $c\cdot p_i^{1/3}$ internal faces.
   Now we have
   $$|V(\tilde{G_t})|^{\frac{1}{3}}\leq (6(u+s) +\sum p_i)^{\frac{1}{3}}\leq (6u+6s)^{\frac{1}{3}}+\sum p_i^{\frac{1}{3}}\leq (11k)^{\frac{1}{3}}+\frac{1}{c}(k-2) \leq \frac{k}{c},$$
   for \(k \geq 1\) and \(c \geq \frac{3}{\sqrt{11}}\)
   which directly implies the desired bound on the number of internal faces in the face cover of $\tilde{G_t}$.
\end{proof}

\subsection{P-node}
\label{sec:pnode}
Since we assume that $G$ has no parallel edges, no leaf node of the \SPQR-tree of $G$ is a P-node. While the graph $G$ has no parallel edges, the application of certain reduction rules may create parallel edges. We deal with those edges as follows.

Let $t$ be a parallel node and let $G_t$ be the enhancement of the graph induced on the sub-\SPQR-tree rooted at $t$. In the following discussion, after applying each of the reduction rules exhaustively the graph obtained from $G_t$ will be denoted by $G'_t$.

\begin{rr}\label{rr:par-real}
    Delete all but one real edge of the skeleton of \(G_t\) from $G_t$.
\end{rr}
\begin{restatable}{proposition}{psafeedgedel}
    \Cref{rr:par-real} creates a nice kernel.
\end{restatable}
\begin{proof}
    Let $G'_t$ be the graph obtained from $G_t$ after exhaustively applying \Cref{rr:par-real}. Since we do not modify the corners or the corner edge, property (K1) trivially holds good for $G'_t$. We claim that property (K2) is satisfied too. Let $F$ be the face cover of $G_t$ with the minimum number of faces over all embeddings. Abusing notation, we denote the embedding realizing $F$ by $G_t$. Let $F'$ be a face cover in the embedding $G'_t$ obtained from $G_t$ by deleting all but one arbitrarily chosen real edge. Note that deleting an edge cannot increase the number of faces in $G_t$ nor can it increase the number of terminals. This implies $|F'| \leq |F|$. To see the converse, let $F'$ be a minimum face cover of the embedding $G'_t$. We modify the embedding $G'_t$ to get an embedding of $G_t$ by adding the real edges back one by one at $\varepsilon$-distance from the rightmost real edge at each step. This operation does not increase the number of terminals. We possibly change at most one face boundary of $F'$ (the one containing the arbitrary real edge on its boundary). However, we do not change the behavior of that face with respect to covering the terminals as a real edge gets replaced by another real edge. Thus, the faces corresponding to faces of $F'$ in $G_t$ still cover all the terminals. Hence $|F|\leq |F'|$. Note that in any of the above operations we do not change the behavior of $F'$ with respect to containing zero, one or two external faces. This implies that $(K3)$ is satisfied.
\end{proof}

Trivially, there can be at most $k$ problematic virtual components in $G_t$. We will next argue that we can bound the number of terminal-free, unproblematic and problematic virtual components linearly in $k$.

\begin{rr}\label{rr:para-bundle}
 Delete all but one terminal-free virtual component from $G_t$. 
\end{rr}

\begin{restatable}{proposition}{psafebundle}\label{prop:safeness-bundle}
    Applying \Cref{rr:para-bundle} creates a nice kernel.
\end{restatable}
\begin{proof}
    Let $G'_t$ be the graph obtained after exhaustively applying \Cref{rr:para-bundle}.
     Note that property (K1) is trivially satisfied by $G'_t$ as we neither modify the corner vertices nor the corner edge. 
    
    To show that property (K2) still holds, we first show that if $\fcn(G_t)\leq k $, then $\fcn(G'_t) \leq k$. Consider an embedding of $G_t$ with a face cover $F$ of size $k$. We obtain $G'_t$ from it by deleting all but one arbitrarily chosen terminal-free virtual component. Let the face cover of $G'_t$ be $F'$. Since deleting virtual components does not increase the number of faces in the P-node, nor does it create extra terminals, $|F'|\leq |F|\leq k$. Also note that if a terminal-free virtual component bordered an external face of $G_t$, this face would not be covering any terminals and can be removed from $F$ without affecting its feasibility. Therefore, we do not change the behavior of the face cover with respect to inclusion of none, one, or both of the external faces. Next, we show that if $\fcn(G'_t)\leq k$ then $ \fcn(G_t)\leq k$. Let $F'$ be a face cover of $G'_t$ of size at most $k$. We add to $G'_t$ all the terminal-free virtual components one by one, embedding them at $\varepsilon$-distance to the right of the rightmost terminal-free virtual component in $G'_t$. We abuse notation and call this embedding $G_t$ as well. This operation changes the boundary of the face of $G'_t$ to the right of the terminal-free component. However, it replaces a terminal-free path between the corners by another terminal-free path. Additionally, it creates new faces, none of which are incident on any terminal (except maybe the corners). More importantly, no new terminals are created and therefore the faces of $F'$ with potentially modified boundaries still form a face cover in $G_t$. Therefore, $\fcn(G_t)\leq \fcn(G'_t)\leq k$.

    Finally, we claim that property (K3) is satisfied. If none of the corners of $G_t$ were terminals, the claim is trivially true. Suppose that at least one corner was a terminal. Let $f$ be the face in $F$ covering the corner and let $f$ belong to a terminal-free component. Note that if $F$ contains any face corresponding to an internal face of the skeleton of the P-node $t$, then such a face would cover the corner vertex. In that case, we can remove $f$ from $F$ without compromising its feasibility. Otherwise, we remove $f$ from $F$ and add a face corresponding to an internal face of the skeleton of the P-node if $f$ is internal and an external face otherwise. This does not uncover any terminal as $f$ belonged to a terminal-free component to begin with. Moreover, the number of external faces in $|F|$ remains the same.
\end{proof}

We claim that after we apply Reduction rules~\ref{rr:par-real} and~\ref{rr:para-bundle} exhaustively, the number of virtual child components of a P-node is bounded linearly in $k$. We say that two virtual child components \emph{share an internal face} if some internal face of the P-node skeleton is formed by merging external faces contained in the respective face covers of the child virtual components. 

\begin{restatable}{proposition}{pnodevircombound}\label{prop:para-bound-vir-comp}
   If there are more than $4k+2$ virtual components in $G'_t$, then there does not exist a face cover of size at most $k$ in $G$. 
\end{restatable}
\begin{proof}
    There can be at most $k$ virtual components that need a private internal face to cover any of their terminals. By \Cref{rr:para-bundle}, there is one edge in the skeleton replacing all the terminal-free virtual components. We restrict our attention to those virtual components whose terminals can be completely covered by their external faces. 
    
    Each internal face of the P-node skeleton contained in its face cover is bounded by two virtual edges. Effectively, each child component corresponding to these edges contributes half of this face to the face cover.
    
    Consider all the unproblematic child components. Each unproblematic child component has exactly one of its external faces in its face cover. Therefore, it contributes only half an internal face to the face cover of the P-node skeleton. Therefore, if there are more than $2k$ of them, then there would be at least $k+1$ faces in the face cover.
    
    Next, consider all the semi-problematic child components. These components either have one of their internal faces in their face cover or both of their external ones. So, they contribute either one internal face of their own or half of two internal faces of the P-node skeleton to the face cover. Hence, there can be at most $k$ of them in a P-node, or else the face cover would have at least $k+1$ faces. Adding up, there are at most $4k+2$ virtual components containing terminals.
\end{proof}     

By the above proposition, the number of virtual edges in the skeleton of a P-node is ${\mathcal{O}}(k)$. Replacing the corresponding virtual components by small graphs in $G'_t$ helps to bound its size. We now apply \Cref{rr:terminal-free} through \Cref{rr:semi-problematic} on the graph $G'_t$ exhaustively. Additionally, we replace every problematic virtual component by its small nice kernel.

\begin{restatable}{proposition}{psmallnicekernel}
    After applying Reduction Rules~\ref{rr:par-real}, \ref{rr:para-bundle} and the Reduction rules~\ref{rr:terminal-free}~to~\ref{rr:semi-problematic}~exhaustively on $G_t$, the resulting graph is a small nice kernel 
\end{restatable}
\begin{proof}
   By Proposition~\ref{prop:para-bound-vir-comp}, there are at most $\mathcal{O}(k)$ virtual child components in a $P$-node. Let $G'_t$ be the graph obtained after applying all the above mentioned reduction rules exhaustively. Let $P_1, \dots P_{\ell}$ be the number of problematic child components in $G_t$ and let $p_i$ be the number of vertices in $P_i$. Each problematic component contributes one internal face to the minimum face cover of $G_t$ and each of the other virtual components contribute at least half of the remaining internal faces. Since, there are a constant number of vertices for each unproblematic and semi problematic component and at least $2(k-\ell)$ of these components, the size of $G'_t$ is $\sum_{i=1}^{\ell}p_i + 2(k-\ell) d$, for some constant $d$.
   Using our inductive hypothesis, property (K4) is satisfied by each of these components. Let the number of internal faces in a minimum face cover of $G'_t$ be denoted by $\fcn^i$.\begin{align*}
       \fcn^i&\geq \left(\sum_{1=1}^{\ell} c \cdot {p_i}^{\frac{1}{3}} \right)+ \left(\sum_{i=1}^{2(k-\ell)}c \cdot d^{\frac{1}{3}}\right)\\& \geq c(\sum_{i=1}^{l}p_i + (2k-\ell) d)^{\frac{1}{3}} \\ & = c(|V(G'_t)|^{\frac{1}{3}})
   \end{align*}
Therefore, $G'_t$ is a small nice kernel. 
\end{proof}

\section{Conclusion}
\label{sec:conclusion}
This paper presents the first polynomial kernel for the \textsc{Face Cover Number} problem, with a cubic dependence on the parameter. For a restricted version of the problem -- where the embedding is provided as part of the input and the terminal set includes all vertices of the graph -- a linear kernel was previously known. A natural direction for future research is to close the gap between the kernel sizes for the embedded and non-embedded variants of the problem.

\vspace{2pt}
\noindent\textbf{Open Question 1.} \textit{Does there exist a linear kernel for \textsc{Face Cover Number}?}

Addressing this question likely necessitates a deeper understanding of the structural properties of 3-connected planar graphs.

Following the fixed-parameter tractable (\fpt) algorithm developed by Bienstock and Monma~\cite{BM88}, there have been no subsequent improvements for the non-embedded case. In particular, no subexponential-time algorithm is currently known for \textsc{Face Cover}, prompting the following question:

\vspace{2pt}
\noindent\textbf{Open Question 2.} \textit{Is there an algorithm that solves \textsc{Face Cover} in $2^{\mathcal{O}(\sqrt{k})}$ time, where $k$ is the face cover number?}

It is important to note that while bidimensionality theory applies to the variant with $T = V(G)$, the general form of the problem does not fall within this framework. Addressing Open Question~2 will therefore require the development of novel algorithmic ideas beyond existing techniques.

\bibliographystyle{plain}
\bibliography{fc-references}

\end{document}